\documentclass[12pt]{article}
\usepackage{amsmath,amssymb,epsf,epsfig,amsthm,times,ifthen,epic}
\usepackage{enumerate}
\usepackage{printtime}
\usepackage{cite}     
\usepackage{fancyheadings}
\usepackage{setspace} 
\usepackage{lastpage} 
\usepackage{hyperref} 

\def\submitteddate{November 18, 2013}

\renewcommand{\baselinestretch}{1.0}
\begin{document}

\newcommand{\creationtime}{\today \ \ @ \theampmtime}

\pagestyle{fancy}
\renewcommand{\headrulewidth}{0cm}
\chead{\footnotesize{Dougherty-Freiling-Zeger}}
\rhead{\footnotesize{\submitteddate}}
\lhead{}
\cfoot{Page \arabic{page} of \pageref{LastPage}} 

\renewcommand{\qedsymbol}{$\blacksquare$} 


\newtheorem{theorem}              {Theorem}     [section]
\newtheorem{lemma}      [theorem] {Lemma}
\newtheorem{corollary}  [theorem] {Corollary}
\newtheorem{proposition}[theorem] {Proposition}
\newtheorem{remark}     [theorem] {Remark}
\newtheorem{algorithm}  [theorem] {Algorithm}
\newtheorem{conjecture} [theorem] {Conjecture}
\newtheorem{example}    [theorem] {Example}

\theoremstyle{definition}         
\newtheorem{definition} [theorem] {Definition}
\newtheorem*{claim}  {Claim}
\newtheorem*{notation}  {Notation}



\newcommand{\alphabet}{\mathcal{A}}
\newcommand{\netI}{\mathcal{N}}
\newcommand{\netII}{\mathcal{N}'}

\newcommand{\Span}[1]{\langle #1 \rangle}
\newcommand{\Z}{\mathbf{Z}}
\newcommand{\Q}{\mathbf{Q}}
\newcommand{\N}{\mathbf{N}}
\newcommand{\R}{\mathbf{R}}
\newcommand{\A}{\mathcal{A}}
\newcommand{\I}{\mathcal{I}}
\newcommand{\V}[2]{\mathrm{V}(#1,#2)}

\newcommand{\yields}{\longrightarrow}
\newcommand{\Matroid}{\mathcal{M}}
\newcommand{\GroundSet}{\mathcal{S}}
\newcommand{\ClosureOperator}{\mathrm{cl}}
\newcommand{\Closure}[1]{\ClosureOperator(#1)}
\newcommand{\IndepSets}{\mathcal{I}}
\newcommand{\MessageSet}{\mu}
\newcommand{\NodeSet}{\nu}
\newcommand{\EdgeSet}{\epsilon}
\newcommand{\SourceMapping}{S}
\newcommand{\ReceiverMapping}{R}
\newcommand{\AssignmentMapping}{a}
\newcommand{\Carry}{c}
\newcommand{\NodeInSymbols} [1]{\mathrm{In}(#1)}
\newcommand{\NodeOutSymbols}[1]{\mathrm{Out}(#1)}
\newcommand{\Vamos}{V\'{a}mos}
\newcommand{\PowerSet}[1]{2^{#1}}

\newcommand{\range}{\mathsf{range}}
\newcommand{\theranksymbol}{\mathsf{rank}}
\newcommand{\Rank}[1]{\mathsf{rank}\left(#1\right)}
\newcommand{\Nullity}[1]{\mathsf{nullity}(#1)}
\newcommand{\Dim}[1]{\mathsf{dim}(#1)}
\newcommand{\Codim}[2]{\mathsf{codim}_{#1}\left(#2\right)}
\newcommand{\Comment}[1]{& [\mbox{from  #1}]}

\newcommand{\Implies}[2]{$H(#2 | #1)=0$}
\renewcommand{\emptyset}{\varnothing} 
\renewcommand{\subset}{\subseteq}     
\newcommand{\FanoNetwork}{\mathcal{N}_F}
\newcommand{\nonFanoNetwork}{\mathcal{N}_{nF}}
\newcommand{\Network}{\mathcal{N}}
\newcommand{\TBA}{*** To Be Added ***}

\newcommand{\ARR}{S}
\newcommand{\Cave}{\mathcal{C}^{\mathrm{average}}}
\newcommand{\Cuni}{\mathcal{C}^{\mathrm{uniform}}}
\newcommand{\CaveLinear} {\Cave_{\mathrm{linear}}}
\newcommand{\CuniLinear} {\Cuni_{\mathrm{linear}}}
\newcommand{\CuniRouting}{\Cuni_{\mathrm{routing}}}
\newcommand{\CaveRouting}{\Cave_{\mathrm{routing}}}
\newcommand{\polytope}{\mathcal{P}}

\newcommand{\LHS}{\mathrm{LHS}}
\newcommand{\RHS}{\mathrm{RHS}}
\newcommand{\Restriction}[2]{#1 \arrowvert #2}


\setcounter{page}{0}

\title{Achievable Rate Regions for Network Coding
\thanks{This work was supported by the 
Institute for Defense Analyses and
the National Science Foundation.\newline
\indent \textbf{R. Dougherty} is with the 
Center for Communications Research, 
4320 Westerra Court, 
San Diego, CA 92121-1969 
(rdough@ccrwest.org).\newline
\indent \textbf{C. Freiling} is with  the 
Department of Mathematics, 
California State University, San Bernardino, 
5500 University Parkway, 
San Bernardino, CA 92407-2397 
(cfreilin@csusb.edu).\newline
\indent \textbf{K. Zeger} is with the 
Department of Electrical and Computer Engineering, 
University of California, San Diego, 
La Jolla, CA 92093-0407 
(zeger@ucsd.edu).
}}

\author{Randall Dougherty, Chris Freiling, and Kenneth Zeger\\ \ }

\date{
\textit{
Submitted: \submitteddate\\
}}

\maketitle
\begin{abstract}
Determining the achievable rate region for networks
using routing, linear coding, or non-linear coding
is thought to be a difficult task in general,
and few are known.
We describe the achievable rate regions for four interesting networks
(completely for three and partially for the fourth).
In addition to the known matrix-computation method for proving
outer bounds for linear coding, we present a new method which
yields actual characteristic-dependent linear rank inequalities
from which the desired bounds follow immediately.
\end{abstract}

\addtocontents{toc}{[This Table of Contents is for draft preparation only. It will be removed prior to submission.]\\\endgraf}
\addtocontents{toc}{\hfill Page\endgraf}
\addcontentsline{toc}{section}{Title, Authors, Abstract}

\thispagestyle{empty}


\newpage
\section{Introduction}

In this paper,
a \textit{network} is a directed acyclic multigraph $G=(V,E)$,
some of whose nodes are information sources or receivers 
(e.g.\ see \cite{Yeung-book}).
Associated with the sources are $m$ generated \textit{messages},
where the $i^{th}$ source message is assumed to be a vector of $k_i$ arbitrary elements of a fixed finite alphabet, 
$\alphabet$, of size at least $2$.
At any node in the network,
each out-edge carries a vector of $n$ alphabet symbols which is a function
(called an \textit{edge function})
of the vectors of symbols carried on the in-edges to the node,
and of the node's message vectors if it is a source.
Each network edge is allowed to be used at most once 
(i.e.\ at most $n$ symbols can travel across each edge).
It is assumed that every network edge is reachable by some source message.
Associated with each receiver are one or more \textit{demands}; each demand is a network message.
Each receiver has \textit{decoding functions} which map the receiver's inputs to vectors of symbols
in an attempt to produce the messages demanded at the receiver.
The goal is for each receiver to deduce its demanded messages from its in-edges and source messages
by having information propagate from the sources through the network.

A $(k_1, \ldots, k_m, n)$ \textit{fractional code} is a collection of edge functions, 
one for each edge in the network,
and decoding functions, 
one for each demand of each node in the network.
A $(k_1, \ldots, k_m, n)$ \textit{fractional solution} is 
a $(k_1, \ldots, k_m, n)$ fractional code which results in every receiver 
being able to compute its demands via its decoding functions,
for all possible assignments of length-$k_i$ vectors over the alphabet to the $i^{th}$ source message, for all $i$.

Special codes of interest include \textit{linear codes},
where the edge functions and decoding functions are linear,
and \textit{routing codes},
where the edge functions and decoding functions simply copy 
specified input components to output components.\footnote{
If an edge function for an out-edge of a node 
depends only on the symbols of a single in-edge of that node,
then, without loss of generality,
we assume that the out-edge simply carries the same vector of symbols (i.e.\ routes the vector) as the in-edge it depends on.
}
Special networks of interest include \textit{multicast} networks,
where there is only one source node and every receiver demands all of the source messages,
and \textit{multiple-unicast} networks,
where each network message is generated by exactly one source node and is demanded by exactly
one receiver node.

For each $i$, the ratio $k_i/n$ can be thought of as the rate at which source $i$ injects
data into the network.
If a network has a $(k_1, \ldots, k_m, n)$ fractional solution over some alphabet,
then we say that $\left(k_1/n, \ldots, k_m/n\right)$ is an \textit{achievable rate vector},
and we define the \textit{achievable rate region} of the network as the following convex hull%
\footnote{
There is some variation in the definition and terminology in the literature.
Some authors use the term 
``capacity region'' or
``rate region''.
Alternative definitions of the region have been defined as the 
topological closure of $\ARR$ 
or without the convex hull.
}
$$\ARR = \mathrm{CHULL} (\{ r\in \Q^m: r \ \mbox{ is an achievable rate vector} \}) .$$
Every vector in the achievable rate region
can be effectively achieved
by time-sharing between two achievable points
(since it is a convex combination of those achievable points).

Determining the achievable rate region of an arbitrary network appears to be a formidable task.
Alternatively, certain scalar quantities that reveal information about the achievable rates are typically studied.
For any $(k_1, \ldots, k_m, n)$ fractional solution, we call the scalar quantity
$$\frac{1}{m}\left(\frac{k_1}{n} + \dots + \frac{k_m}{n}\right)$$ an \textit{achievable average rate} of the network.
We define the \textit{average coding capacity} of a network 
to be the supremum of all achievable average rates, namely
$$\Cave = \sup \left\{ \frac{1}{m}\sum_{i=1}^m r_i : (r_1, \dots, r_m) \in \ARR \right\}.$$
Similarly, for any $(k_1, \ldots, k_m, n)$ fractional solution, we call the scalar quantity
$$\min\left(\frac{k_1}{n}, \dots,  \frac{k_m}{n}\right)$$ an \textit{achievable uniform rate} of the network.
We define the \textit{uniform coding capacity} of a network 
to be the supremum of all achievable uniform rates, namely
$$\Cuni = \sup \left\{ \min(r_1, \dots, r_m) : (r_1, \dots, r_m) \in \ARR \right\}.$$
Note that for any $r\in\ARR$ and $r'\in\R^m$,
if each component of $r'$ 
is nonnegative, rational,
and less than or equal to the corresponding component of $r$,
then $r'\in\ARR$.
In particular, if $(r_1, \dots, r_m)\in\ARR$ and $r_i = \displaystyle\min_{1\le j\le m} r_j$,
then $(r_i, r_i, \ldots, r_i)\in \ARR$,
which implies
$$\Cuni = \sup \left\{ r_i : (r_1, \dots, r_m) \in\ARR,\ \ r_1 = \cdots = r_m 
\right\}.$$
In other words, all messages can be restricted to having the same 
dimension $k_1 = \dots = k_m$ when considering $\Cuni$.
Also, note that $$\Cuni \le \Cave.$$
The quantities $\Cave$ and $\Cuni$ are attained by points on the boundary of $\ARR$.
It is known that not every network has a uniform coding
capacity which is an achievable uniform rate
\cite{Dougherty-Freiling-Zeger05-Unachievability}.

If a network's edge functions are restricted to purely routing functions,
then we write the capacities as $\CaveRouting$ and $\CuniRouting$,
and refer to them as the
\textit{average routing capacity} and \textit{uniform routing capacity}, respectively.
Likewise, for solutions using only linear edge functions,
we write
$\CaveLinear$ and $\CuniLinear$ 
and refer to them as the
\textit{average linear capacity} and \textit{uniform linear capacity}, respectively.

Given random variables $x_1, \dots x_i$ and $y_1, \dots, y_j$,
we write $x_1, \dots x_i \yields y_1, \dots, y_j$ to mean that
$y_1, \dots, y_j$ are deterministic functions of $x_1, \dots x_i$.
We say that $x_1, \dots x_i$ \textit{yield} $y_1, \dots, y_j$.

In this paper,
we study four specific networks,
namely
the Generalized Butterfly network,
the Fano network,
the non-Fano network,
and
the \Vamos{} network.
The last three of these networks were shown to be matroidal in  
\cite{Dougherty-Freiling-Zeger07-NetworksMatroidsNonShannon} 
and various capacities of these networks have been computed.
However,
the full achievable rate regions of these networks
have not been previously determined,
to the best of our knowledge.
Some other work on achievable rates and capacities has been done in
~\cite{
Chan-Grant-CapacityRegions,
Kim-CapacityRegions,
Yan-Yeung-Zhang-ISIT07
}.

The Generalized Butterfly network 
(studied in Section~\ref{sec:Butterfly} and illustrated in Figure~\ref{fig:generalized-butterfly})
has the same topology as the usual Butterfly network \cite{Ahlswede-Cai-Li-Yeung-IT-Jul00},
but instead of one source at each of nodes $n_1$ and $n_2$,
there are two sources at each of these nodes.
For each of the source nodes,
one of it's source messages is demanded by receiver $n_5$
and the other by receiver $n_6$.
The usual Butterfly network is the special case when messages $a$ and $d$
do not exist 
(or are just not demanded by any receiver).
A large majority of network coding publications mention in some context
the Butterfly network, so it plays an important role in the field.

The Fano network
(studied in Section~\ref{sec:Fano} and illustrated in Figure~\ref{fig:Fano-network}) 
and the non-Fano network
(studied in Section~\ref{sec:non-Fano} and illustrated in Figure~\ref{fig:non-Fano-network}) 
were used in \cite{Dougherty-Freiling-Zeger05-Unachievability} 
as components of a larger network to demonstrate the unachievability of
network coding capacity.
Specifically, 
in \cite{Dougherty-Freiling-Zeger05-Unachievability} 
the Fano network was shown to be solvable if and only if the alphabet size is a power of $2$
and the non-Fano network was shown to be solvable if and only if the alphabet size is odd.
In \cite{Dougherty-Freiling06-Nonreversibility},
the Fano and non-Fano networks were used to build a solvable multicast network 
whose reverse 
(i.e. all edge directions change, and sources and receivers exchange roles)
was not solvable,
in contrast to the case of linear solvability,
where reversals of linearly solvable multicast networks were previously known to be linearly solvable
\cite{
Koetter-etal-reversible,
Koetter-DIMACS03,
Riis-reversibility-preprint
}.
In \cite{Dougherty-Freiling-Zeger04-Insufficiency},
the Fano and non-Fano networks were used to construct a network
which disproved a previously published conjecture asserting
that all solvable networks are vector linearly solvable
over some finite field and some vector dimension.

The \Vamos{} network 
(studied in Section~\ref{sec:Vamos} and illustrated in Figure~\ref{fig:Vamos-network}) 
was used in
\cite{Dougherty-Freiling-Zeger07-NetworksMatroidsNonShannon}
to demonstrate that non-Shannon-type information inequalities
could yield upper bounds on network coding capacity which
are tighter than the tightest possible bound theoretically
achievable using only Shannon-type information inequalities.
Here we completely determine the routing and linear rate regions
for the \Vamos{} network, but only give partial results
for the non-linear rate region (which indicate that it could
be quite complicated).

Finally, we present a new method for proving bounds on achievable
rate regions for linear coding, which actually produces
explicit linear rank inequalities which directly imply
the desired bounds.

\newpage
\section{Generalized Butterfly network}
\label{sec:Butterfly}

\begin{figure}[h]
\begin{center}
\leavevmode
\hbox{\epsfxsize=0.3\textwidth\epsffile{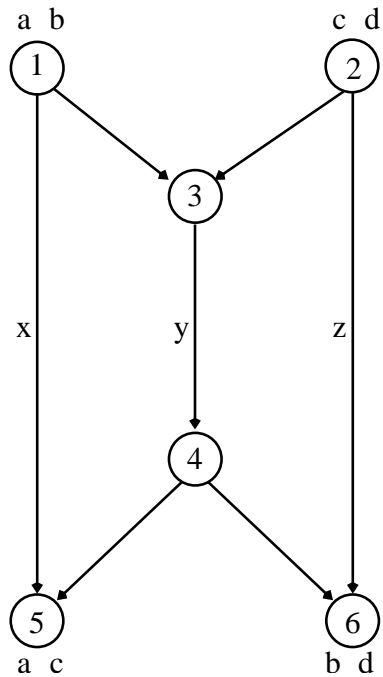}}
\end{center}
\caption{The Generalized Butterfly network. 
Source node $n_1$ generates messages $a$ and $b$,
and
source node $n_2$ generates messages $c$ and $d$.
Receiver node $n_5$ demands messages $a$ and $c$,
and
receiver node $n_6$ demands messages $b$ and $d$.
The symbol vectors carried on edges $e_{1,5}$, $e_{2,4}$, and $e_{3,6}$ are denoted $x$, $y$, and $z$, respectively.
}
\label{fig:generalized-butterfly}
\end{figure}
\begin{theorem}
The achievable rate regions 
for either linear or non-linear coding are the same
for the Generalized Butterfly network
and are equal to the closed polytope in $\R^4$ whose faces lie on the $9$ planes:
\begin{align*}
r_a &= 0\\
r_b &= 0\\
r_c &= 0\\
r_d &= 0\\
r_b &= 1\\
r_c &= 1\\
r_a + r_b + r_c &= 2\\
r_b + r_c + r_d &= 2\\
r_a + r_b + r_c + r_d &= 3
\end{align*}
and whose vertices are the $14$ points:
\begin{align*}
&(0,0,0,0) &&(0,0,0,2) &(2,0,0,0) &&(0,1,0,0)\ \\
&(0,0,1,0) &&(2,0,0,1) &(1,0,0,2) &&(0,0,1,1)\ \\
&(1,1,0,0) &&(1,0,1,1) &(1,1,0,1) &&(0,1,1,0)\ \\
&(0,1,0,1) &&(1,0,1,0).          
\end{align*}
Furthermore,
the coding capacity and linear coding capacity are given by:
\begin{align*}
\Cuni &= \CuniLinear = 2/3\\
\Cave &= \CaveLinear \, = 3/4.
\end{align*}
\label{thm:1}
\end{theorem}

\begin{proof}
Consider a network solution over an alphabet $\alphabet$ and
denote the source message dimensions by $k_a$, $k_b$, $k_c$, and $k_d$,
and the edge dimensions by $n$.
Let each source be a random variable whose components are independent and uniformly distributed over $\alphabet$.
Then the solution must satisfy the following inequalities:
\begin{align}
                  k_a &\ge 0  \label{eq:1}\\
                  k_b &\ge 0  \label{eq:2}\\
                  k_c &\ge 0  \label{eq:3}\\
                  k_d &\ge 0  \label{eq:4}\\
                  k_b &= H(b) = H(y|a,c,d)    \le n \label{eq:5}\\ 
                  k_c &= H(c) = H(y|a,b,d)    \le n \label{eq:6}\\ 
      k_a + k_b + k_c &= H(a,b,c) = H(x,y|d)\notag\\
                      &\le H(x,y) \le 2n \label{eq:7}\\
      k_b + k_c + k_d &= H(b,c,d) = H(y,z|a) \notag\\
                      &\le H(y,z) \le 2n \label{eq:8}\\
k_a + k_b + k_c + k_d &= H(a,b,c,d) = H(x,y,z) \notag\\
                      &\le 3n. \label{eq:9}
\end{align}
\eqref{eq:1}--\eqref{eq:4} are trivial;
\eqref{eq:5} follows because $c,d,y \yields y,z \yields b,d$ (at node $n_6$), 
             and therefore $a,c,d,y \yields a,b,c,d$ 
             and thus $H(a,b,c,d) = H(a,c,d,y)$;
similarly for \eqref{eq:6};
\eqref{eq:7} follows because $x,y \yields a,c$   (at node $n_5$),
                             $c,d,y \yields b,d$ (at node $n_6$),
             and therefore $d,x,y \yields a,c,d,y \yields a,b,c,d$
             and thus $H(a,b,c,d) = H(d,x,y)$;
similarly for \eqref{eq:8};
\eqref{eq:9} follows because $x,y,z \yields a,b,c,d$ (at nodes $n_5$ and $n_6$).
Dividing each inequality in \eqref{eq:1}--\eqref{eq:9} by $n$ gives the $9$ bounding hyperplanes
stated in the theorem.

Let $r_a = k_a/n$, $r_b = k_b/n$, $r_c= k_c/n$, and $r_d = k_d/n$,
and let $\polytope$ denote the polytope in $\R^4$ consisting of all
$4$-tuples $(r_a, r_b, r_c, r_d)$ satisfying \eqref{eq:1}--\eqref{eq:9}.
Then \eqref{eq:1}--\eqref{eq:4} and \eqref{eq:9} ensure that $\polytope$ is bounded.
One can easily calculate that each point in $\R^4$ that 
satisfies some independent set of four of the inequalities
\eqref{eq:1}--\eqref{eq:9} with equality
and also satisfies the remaining five inequalities
must be one of the $14$ points stated in the theorem.
Now we show that all $14$ such points do indeed lie in the achievable rate
region,
and therefore their convex hull equals the achievable rate
region.
The following $5$ points are achieved
by taking $n = 1$
with the following codes
over any field
(where, if $k_a = 2$, the two components of $a$ are denoted $a_1$ and $a_2$):
\begin{align*}
(2,0,0,1)&\!:\ \ x=a_1,\ y=a_2,\ z=d \\
(1,0,0,2)&\!:\ \ x=a,\ y=d_1,\ z=d_2 \\
(1,0,1,1)&\!:\ \ x=a,\ y=c,\ z=d \\
(1,1,0,1)&\!:\ \ x=a,\ y=b,\ z=d \\
(0,1,1,0)&\!:\ \ x=b,\ y=b+c,\ z=c
\end{align*}
and the remaining $9$ points are achieved by fixing certain messages to be $0$.

Since the above codes are all linear, 
the achievable rate regions for linear and non-linear
codes are the same.

By \eqref{eq:9}, we have $\Cave \le 3/4$,
and this upper bound is achievable by routing using the code
given above for the point $(2,0,0,1)$,
namely taking $x=a_1$, $y=a_2$, and $z=d$.
By \eqref{eq:8}, we have $\Cuni \le 2/3$;
since
\begin{align*} 
(2/3)(1,1,1,1) &= (1/3)(1,0,1,1)\\
&\ \  + (1/3)(1,1,0,1)\\
&\ \  + (1/3)(0,1,1,0)
\end{align*}
the upper bound of $2/3$ is achievable by a convex combination of the linear codes 
given above for the points $(1,0,1,1)$, $(1,1,0,1)$, and $(0,1,1,0)$,
as follows.
Take $k=2$ and $n=3$ and use the (linear) code determined by:
\begin{align*}
x &= (a_1, a_2, b_2)\\
y &= (c_1, b_1, b_2 + c_2)\\
z &= (d_1, d_2, c_2).
\end{align*}
\end{proof}

\begin{theorem}
The achievable rate region 
for routing
for the Generalized Butterfly network
is the closed polytope in $\R^4$ bounded by the $9$ planes in
Theorem~\ref{thm:1}
together with the plane
\begin{align*}
r_b + r_c &= 1
\end{align*}
and whose vertices are the $13$ points:
\begin{align*}
&(0,0,0,0) &&(0,0,0,2) &(2,0,0,0)  &&(0,1,0,0)\\
&(0,1,0,1) &&(0,0,1,0) &(2,0,0,1)  &&(1,0,0,2)\\ 
&(0,0,1,1) &&(1,0,1,0) &(1,1,0,0)  &&(1,0,1,1)\\
&(1,1,0,1).          
\end{align*}
Furthermore,
the routing capacities are given by:
\begin{align*}
\CuniRouting &= 1/2\\
\CaveRouting &= 3/4.
\end{align*}
\label{thm:2}
\end{theorem}

\begin{proof}
With routing,
in addition to the inequalities \eqref{eq:1}--\eqref{eq:9},
a solution must also satisfy
\begin{align}
k_b + k_c &\le n
\label{eq:10}
\end{align}
since all of the components of messages $b$ and $c$ must be carried by the edge labeled $y$.
One can show that each point in $\R^4$ that 
satisfies with equality some independent set of four of the inequalities
\eqref{eq:1}--\eqref{eq:9} and \eqref{eq:10} 
and also satisfies the remaining six inequalities
must be one of the $13$ points stated in this theorem
(i.e. $13$ of the $14$ points stated in Theorem~\ref{thm:1}
by excluding the point $(0,1,1,0)$).
The proof of Theorem~\ref{thm:1} showed that all vertices of $\polytope$ except $(0,1,1,0)$
were achievable using routing.

By \eqref{eq:10}, we have $\CuniRouting \le 1/2$,
and this upper bound is achievable,
for example,
by taking a convex combination of codes that achieve
$(1,0,1,0)$ and $(0,1,0,1)$,
as follows.
Take $k=1$ and $n=2$ and use the routing code determined by:
\begin{align*}
x &= (0, a)\\
y &= (b, c)\\
z &= (d, 0).
\end{align*}
The capacity $\CaveRouting = 3/4$ follows immediately from the proof of Theorem~\ref{thm:1}.
\end{proof}

\newpage
\section{Fano network}
\label{sec:Fano}

\begin{figure}[h]
\begin{center}
\leavevmode
\hbox{\epsfxsize=0.3\textwidth\epsffile{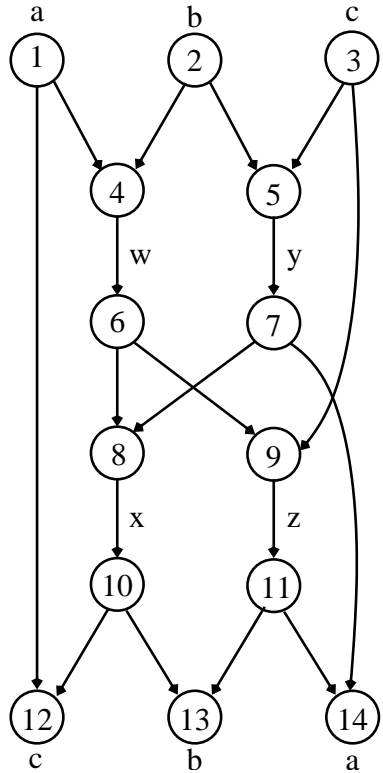}}
\end{center}
\caption{The Fano network. 
Source nodes $n_1$, $n_2$, and $n_3$ generate messages $a$, $b$, and $c$, respectively.
Receiver nodes $n_{12}$, $n_{13}$, and $n_{14}$ demand messages $c$, $b$, and $a$, respectively.
The symbol vectors carried on edges $e_{4,6}$, $e_{8,10}$, $e_{5,7}$, $e_{9,11}$ are labeled
as $w$, $x$, $y$, and $z$, respectively.
}
\label{fig:Fano-network}
\end{figure}
\begin{theorem}
The achievable rate regions 
for either linear coding
over any finite field alphabet of even characteristic
or non-linear coding are the same
for the Fano network
and are equal to the closed polyhedron in $\R^3$ whose faces lie on the $7$ planes 
(see Figure~\ref{fig:Fano-coding-rate-region}):
\begin{align*}
r_a &= 0\\ 
r_b &= 0\\ 
r_c &= 0\\
r_a &= 1\\
r_c &= 1\\
r_b + r_c &= 2\\
r_a + r_b &= 2   
\end{align*}
and whose vertices are the $8$ points:
\begin{align*}
&(0,0,0) &&(0,0,1) &(1,0,0) &&(0,2,0)\ \\
&(0,1,1) &&(1,0,1) &(1,1,0) &&(1,1,1).
\end{align*}
\label{thm:3}
\end{theorem}

\begin{proof}
Consider a network solution over an alphabet $\alphabet$ and
denote the source message dimensions by $k_a$, $k_b$, and $k_c$,
and the edge dimensions by $n$.
Let each source be a random variable whose components are independent and uniformly distributed over $\alphabet$.
Then the solution must satisfy the following inequalities:
\begin{align}
                  k_a &\ge 0                      \label{eq:11}\\
                  k_b &\ge 0                      \label{eq:12}\\
                  k_c &\ge 0                      \label{eq:12b}\\
                  k_a &= H(a) = H(z|b,c)   \le H(z) \le n \label{eq:13}\\ 
                  k_c &= H(c) = H(y|a,b)   \le H(y) \le n \label{eq:14}\\ 
            k_b + k_c &= H(b,c) = H(x,z|a) \le H(x,z) \le 2n \label{eq:15}\\
            k_a + k_b &= H(a,b) = H(x,z|c) \le H(x,z) \le 2n. \label{eq:16}
\end{align}
\eqref{eq:11}--\eqref{eq:12b} are trivial;
\eqref{eq:13} follows because $z,b,c \yields z,y \yields a$ (at node $n_{14}$), 
                              so $z,b,c \yields a,b,c$ 
                              and thus $H(z,b,c) = H(a,b,c)$;
\eqref{eq:14} follows because $a,b,y \yields a,w,y \yields a,x \yields c$ (at node $n_{12}$), 
                              so $a,b,y \yields a,b,c$ 
                              and thus $H(a,b,y) = H(a,b,c)$;
\eqref{eq:15} follows because $a,x,z \yields a,b,c$ (at nodes $n_{12}$ and $n_{13}$) and thus $H(a,x,z)=H(a,b,x)$;
\eqref{eq:16} follows from: $x,z\yields b$ (at node $n_{13}$),
                            $b,c\yields y$ (at node $n_{5}$),
                            $x,z,c\yields z,b,c \yields y,z,b,c \yields a,b,c$,
                            so $H(x,z,c) = H(a,b,c)$.
Dividing each inequality in \eqref{eq:11}--\eqref{eq:16} by $n$ gives the $7$ bounding planes
stated in the theorem.

Let $r_a = k_a/n$, $r_b = k_b/n$, and $r_c= k_c/n$,
and let $\polytope$ denote the polygon in $\R^3$ consisting of all
$3$-tuples $(r_a, r_b, r_c)$ satisfying \eqref{eq:11}--\eqref{eq:16}.
Then $\polytope$ is bounded by \eqref{eq:11}--\eqref{eq:16}.
One can easily calculate that each point in $\R^3$ that 
satisfies some set of three of the inequalities
\eqref{eq:11}--\eqref{eq:16} with equality
and also satisfies the remaining four inequalities
must be one of the $8$ points stated in the theorem.
Now we show that all $8$ such points do indeed lie in $\polytope$.
The following $5$ points are seen to lie in $\polytope$ 
by taking $n = 1$
and the following codes over any even-characteristic finite field:
\begin{align*}
(0,1,1)&\!:\ \ x=y=c,\ w=z=b \\
(1,0,1)&\!:\ \ x=y=c,\ w=z=a \\
(1,1,0)&\!:\ \ x=y=b,\ w=z=a \\
(0,2,0)&\!:\ \ x=y=b_1,\ w=z=b_2 \\
(1,1,1)&\!:\ \ w=a+b,\ y=b+c,\ x=a+c,\ z=a+b+c
\end{align*}
and the remaining $3$ points are achieved by fixing certain messages to be $0$
(note that the codes for $(0,1,1)$, $(1,0,1)$, and $(1,1,0)$ can 
be obtained from the linear code for $(1,1,1)$ 
but we gave routing solutions for them here).

Since the above codes are all linear, 
the achievable rate regions for linear and non-linear
codes are the same.
\end{proof}

It was shown in \cite{Dougherty-Freiling-Zeger04-Insufficiency}
that for the Fano network,
$\Cave = \Cuni = 1$
and 
$\CuniLinear = 1$   for all even-characteristic fields
and
$\CuniLinear = 4/5$ for all odd-characteristic fields.
The calculation of $\CuniLinear = 4/5$ in \cite{Dougherty-Freiling-Zeger04-Insufficiency}
required a rather involved computation.
We now extend that computation to give the following theorem.

\begin{figure}[h]
\begin{center}
\leavevmode
\hbox{\epsfxsize=0.5\textwidth\epsffile{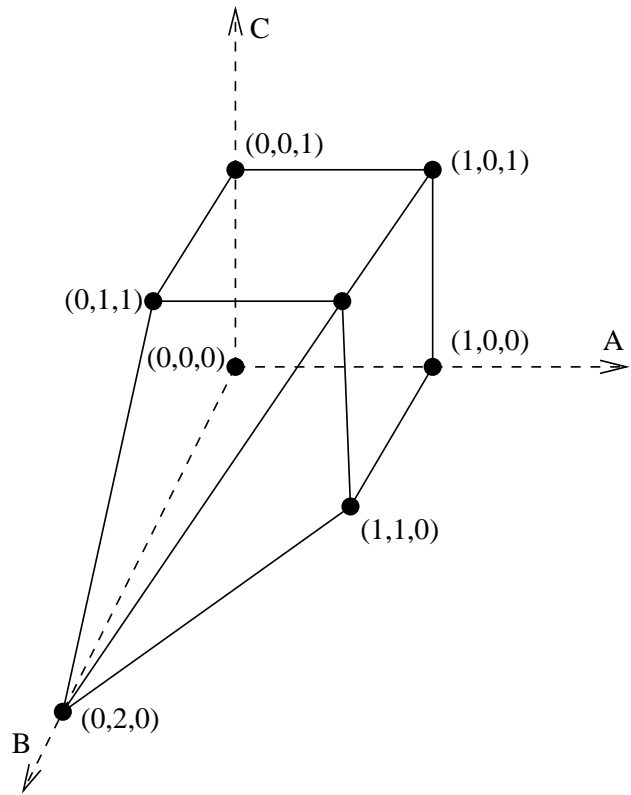}}
\end{center}
\caption{The achievable coding rate region for the Fano network
is a $7$-sided polyhedron with 8 vertices.
}
\label{fig:Fano-coding-rate-region}
\end{figure}
\begin{theorem}
The achievable rate region
for linear coding
over any finite field alphabet of odd characteristic
for the Fano network
is equal to the closed polyhedron in $\R^3$ whose faces lie on the $8$ planes 
(see Figure~\ref{fig:Fano-coding-even-rate-region}):
\begin{align*}
r_a &= 0\\ 
r_b &= 0\\ 
r_c &= 0\\
r_a &= 1\\ 
r_c &= 1\\
r_a + 2r_b + 2r_c &= 4\\
2r_a + r_b + 2r_c &= 4\\
2r_a + 2r_b + r_c &= 4
\end{align*}
and whose vertices are the $10$ points:
\begin{align*}
&(0,0,0) &&(0,0,1) &(1,0,0) &&(0,2,0)\ \\
&(0,1,1) &&(1,0,1) &(1,1,0) &&\ \\
&(2/3,2/3,1) &&(1,2/3,2/3) &(4/5,4/5,4/5).
\end{align*}
\label{thm:7}
\end{theorem}

\begin{proof} 
In addition to satisfying the conditions \eqref{eq:11}--\eqref{eq:16},
the solution must satisfy the following inequalities:
\begin{align}
             k_a + 2k_b + 2k_c \le 4n \label{eq:01}\\
            2k_a +  k_b + 2k_c \le 4n \label{eq:02}\\
            2k_a + 2k_b +  k_c \le 4n \label{eq:03}
\end{align}
The proofs of these inequalities are given in Section~\ref{sec:FanoProof},
and an alternate proof of~\eqref{eq:02} is given in Section~\ref{sec:FanoProof2}.

A straightforward argument as in previous theorems shows that the vertices
of the (bounded) region specified by inequalities \eqref{eq:11}--\eqref{eq:14}
and \eqref{eq:01}--\eqref{eq:03} (inequalities \eqref{eq:15} and \eqref{eq:16}
are now redundant) are the ten vertices listed in the theorem.  For the first
seven of these, the codes given in Theorem~\ref{thm:3} work here as well;
the remaining points are attained by the following three codes
(the last of which was given in~\cite{Dougherty-Freiling-Zeger04-Insufficiency}):
\begin{align*}
(1,2/3,2/3)&\!:\ \ n=3, \\
           & w=(a_1+b_1,a_2+b_2,a_3)\\
           & x=(a_1-c_1,a_2-c_2,a_2+b_2)\\
           & y=(b_1+c_1,b_2+c_2,b_1)\\ 
           & z=(a_1+b_1-c_1,a_2+b_2+c_2,a_3)\\
(2/3,2/3,1)&\!:\ \ n=3, \\
           & w=(a_1+b_1,a_2+b_2,b_2)\\ 
           & x=(a_1-c_1,a_2-c_2,c_3)\\
           & y=(b_1+c_1,b_2+c_2,c_3)\\ 
           & z=(a_1+b_1-c_1,a_2-b_2-c_2,c_1)\\
(4/5,4/5,4/5)&\!:\ \ n=5, \\
           & w=(a_1+b_1,a_2+b_2,a_3+b_3,a_4+b_4,b_1+b_4)\\ 
           & x=(c_1+a_1,c_2+a_2,c_3-a_3,c_4-a_4,a_3+b_3)\\
           & y=(c_1-b_1,c_2-b_2,c_3+b_3,c_4+b_4,b_2)\\ 
           & z=(a_1+b_1+c_1,a_2+b_2+c_2,a_3+b_3+c_3,a_4+b_4+c_4,b_1+b_4+c_4)
\end{align*}
\end{proof}

\begin{figure}[h]
\begin{center}
\leavevmode
\hbox{\epsfxsize=0.5\textwidth\epsffile{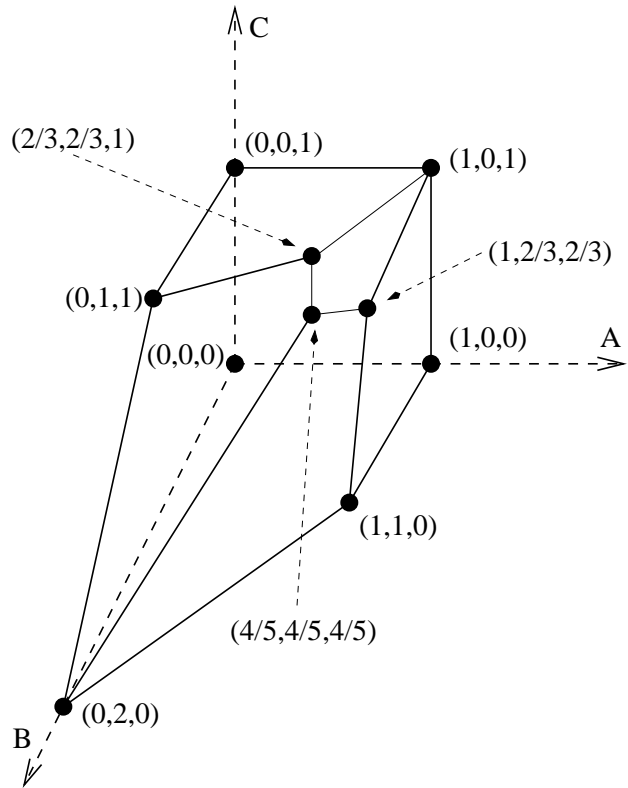}}
\end{center}
\caption{The achievable linear coding rate region over even-characteristic finite fields
for the Fano network is a $8$-sided polyhedron with 8 vertices.
}
\label{fig:Fano-coding-even-rate-region}
\end{figure}
\begin{theorem}
The achievable rate region
for routing
for the Fano network
is the closed polyhedron in $\R^3$ whose faces lie on the $6$ planes 
(see Figure~\ref{fig:Fano-routing-rate-region}):
\begin{align*}
r_a &= 0\\
r_b &= 0\\ 
r_c &= 0\\
r_a &= 1\\
r_c &= 1\\ 
r_a + r_b + r_c &= 2
\end{align*}
and whose vertices are the $7$ points:
\begin{align*}
&(0,0,0) &&(0,0,1) &(1,0,0) \ &&(0,2,0)\ \\
&(0,1,1) &&(1,0,1) &(1,1,0).
\end{align*}
\label{thm:4}
\end{theorem}

\begin{proof} %
With routing,
in addition to the inequalities \eqref{eq:11}--\eqref{eq:16},
a solution must also satisfy
\begin{align}
k_a + k_b + k_c &\le 2n
\label{eq:17}
\end{align}
since all of the components of messages $a$, $b$, and $c$ must be carried by the edges labeled $x$ and $z$.
One can easily check that the extreme points of the new region
with the inequality~\eqref{eq:17} added 
are the $7$ points stated in this theorem
(i.e., the points stated in Theorem~\ref{thm:3}
excluding the point $(1,1,1)$); see figure~\ref{fig:Fano-routing-rate-region}.
The proof of Theorem~\ref{thm:3} showed that all vertices of $\polytope$
other than $(1,1,1)$
were achievable using routing.

\end{proof}

\begin{figure}[h]
\begin{center}
\leavevmode
\hbox{\epsfxsize=0.5\textwidth\epsffile{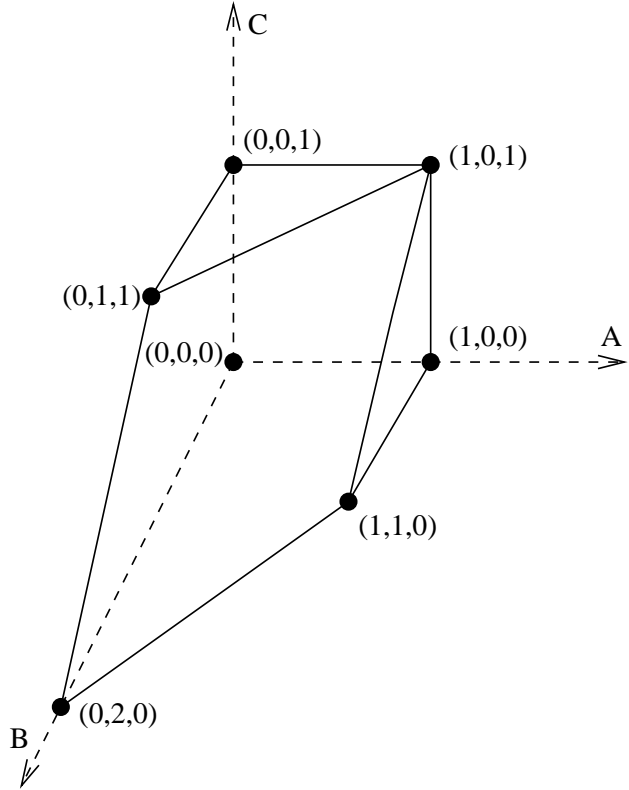}}
\end{center}
\caption{The achievable routing rate region for the Fano network
is a $6$-sided polyhedron with 7 vertices.
}
\label{fig:Fano-routing-rate-region}
\end{figure}

\newcommand{\TwoByOne}[2]{ \left[ \begin{array}{c} #1 \\ #2 \end{array} \right] }
\newcommand{\ThreeByOne}[3]{ \left[ \begin{array}{c} #1 \\ #2 \\ #3 \end{array} \right] }
\newcommand{\MatTwoTwo}[4]{\left[ \begin{array}{cc} #1 & #2 \\ #3 & #4 \end{array} \right]}
\newcommand{\vv}{\mathbf{v}}

\newpage
\section{Proofs of remaining bounds for the Fano network}
\label{sec:FanoProof}

For the case of linear coding over a finite field of odd characteristic,
we want to prove the bounds:
\begin{align}
             k_a + 2k_b + 2k_c \le 4n \label{eq:B122}\\
            2k_a +  k_b + 2k_c \le 4n \label{eq:B212}\\
            2k_a + 2k_b +  k_c \le 4n. \label{eq:B221}
\end{align}
We will do this by following and extending the arguments from
Section~IV of \cite{Dougherty-Freiling-Zeger04-Insufficiency},
with minor modifications needed because we now have separate
source message dimensions $k_a,k_b,k_c$ instead of a
single message dimension $k$.

We already have the bounds $k_a \le n$ and $k_c \le n$
(but we do \textit{not} necessarily have $k_b \le n$).
Therefore, we can think of the length-$n$ symbol vectors
$w$ and $z$ (referred to in~\cite{Dougherty-Freiling-Zeger04-Insufficiency}
as $e_{13,17}$ and $e_{22,30}$) as coming in two parts, one
of length $k_a$ and one of length $\delta_a = n - k_a$.
Similarly, we can think of the symbol vectors
$x$ and $y$ (referred to in~\cite{Dougherty-Freiling-Zeger04-Insufficiency}
as $e_{21,29}$ and $e_{14,18}$) as coming in two parts, one
of length $k_c$ and one of length $\delta_c = n - k_c$.
In order to consider what happens to these parts separately,
we decompose each of the transition matrices $M_i$
from~\cite{Dougherty-Freiling-Zeger04-Insufficiency}
in the form
\begin{align*} 
M_i &=
\MatTwoTwo{R_i}{S_i}{T_i}{U_i}
\end{align*}
where the submatrices $R_i,S_i,T_i,U_i$ are of appropriate sizes
(or are omitted altogether if appropriate).  For instance,
for $i=2$ we have that $R_2$ is $k_a \times k_b$, $T_2$ is
$\delta_a \times k_b$, and $S_2$ and $U_2$ are omitted; for
$i=5$ we have that $R_5$ is $k_c \times k_a$, $S_5$ is
$k_c \times \delta_a$, $T_5$ is $\delta_c \times k_a$,
and $U_5$ is $\delta_c \times \delta_a$.

We can now follow the arguments on pages 2752--2755
of~\cite{Dougherty-Freiling-Zeger04-Insufficiency}
and verify that they apply in this new
context with no further changes.  In particular, the
following formulas from pages 2754 and~2755
of~\cite{Dougherty-Freiling-Zeger04-Insufficiency}
still hold:
\begin{align}
&(U_7 + T_8 S_5) T_2 b + T_8 R_5 R_2 b,
\ T_3 b
\yields\notag\\
&(I + R_8 R_5) R_2 b + (S_7 + R_8 S_5) T_2 b
\label{eq:R6.1}
\end{align}
and
\begin{align}
&T_5 a + T_5 R_2 b + U_5 T_2 b + U_6 T_3 b,\notag\\
&a + R_2 b + S_7 T_2 b - R_8 R_5 a,\notag\\
& U_7 T_2 b - T_8 R_5 a\notag\\
&\yields
b.
\label{eq:R6.2}
\end{align}
Since the field has odd characteristic, we can let
$a' = a + 2^{-1}R_2 b$
and then rewrite \eqref{eq:R6.2} in the following form:
\begin{align}
&T_5 a' + 2^{-1}T_5 R_2 b + U_5 T_2 b + U_6 T_3 b,\notag\\
&(I - R_8 R_5) a' + 2^{-1}((I + R_8 R_5) R_2 b\notag\\
&\ \  + (S_7 + R_8 S_5) T_2 b + (S_7 - R_8 S_5) T_2 b ),\notag\\
&U_7 T_2 b + 2^{-1}T_8 R_5 R_2 b - T_8 R_5 a'\notag\\
&\yields
b.
\label{eq:R6.3}
\end{align}
Note that $a'$ has $k_a$ independent components and is independent
of $b$, just like $a$ is, because $a',b \yields a,b$.

The three vectors on the left-hand side of \eqref{eq:R6.2} have
respective dimensions $\delta_c$, $k_a$, and $\delta_a$; these add up
to $2n - k_c$.  From these vectors we can compute all of $b$
by~\eqref{eq:R6.2}, and then we can also reconstruct some information
about $a$, namely
$(I - R_8 R_5)a$ from the second of the three vectors and
$T_8 R_5 a$ from the third vector.  (We can also get $T_5 a$ from the
first vector, but this will not be used below.)  This gives a total
of
\begin{align*}
k_b + \Rank{\TwoByOne{I - R_8 R_5}{T_8 R_5}}
\end{align*}
independent components reconstructed from these three vectors,
so we must have
\begin{align}
k_b + \Rank{\TwoByOne{I - R_8 R_5}{T_8 R_5}} \le 2n - k_c.
\label{eq:R10.1}
\end{align}

Now, using \eqref{eq:R6.1}, we see that
\begin{align}
T_2 b,
T_3 b,
T_8 R_5 R_2 b
\yields
(I + R_8 R_5)R_2 b.
\label{eq:R10.7}
\end{align}
But we can add $(I + R_8 R_5)R_2 b$ and $(I - R_8 R_5)R_2 b$
to get $2 R_2 b$, which yields $R_2 b$ because the field has
odd characteristic.  And \eqref{eq:R6.2} implies
\begin{align}
a,
T_2 b,
T_3 b,
R_2 b
\yields
a,b.
\label{eq:R10.8}
\end{align}
Putting these together, we get
\begin{align*}
a,
T_2 b,
T_3 b,
\TwoByOne{I - R_8 R_5}{T_8 R_5} R_2 b
\yields
a,b.
\end{align*}
Now, using \eqref{eq:R10.1} and the known sizes of the vectors
$a$, $T_2 b$, and $T_3 b$, we get the inequality
\begin{align*}
k_a + n - k_a + n - k_c + 2n - k_c - k_b &\ge k_a + k_b,
\end{align*}
which reduces to \eqref{eq:B122}.

Using \eqref{eq:R6.1} and \eqref{eq:R6.3} together, we get
\begin{align*}
a',\ 
T_2b,\ 
T_3 b,\ 
T_8 R_5 R_2 b,\ 
T_5 R_2 b\ 
&\yields
a', b \\
&\yields
a, b,
\end{align*}          
yielding the inequality
\begin{align*}
k_a + n - k_a + n - k_c + n - k_a + n - k_c \ge k_a + k_b,
\end{align*}          
which is \eqref{eq:B212}.

For the remaining inequality \eqref{eq:B221}, we will use the
following fact: if $M$ is a $k\times k$ matrix and $N$ is
an $r\times k$ matrix, then
\begin{align}
&\Rank{\TwoByOne{M}{N}} +
\Rank{\TwoByOne{M-I}{N}}\notag\\
&\ \  +
\Rank{\TwoByOne{M+I}{N}}\notag\\
&\ge
2k + \Rank{N}.
\label{eq:R9.1}
\end{align}
Since $1\ne -1$ in a field of odd characteristic, \eqref{eq:R9.1}
is a special case of:

\begin{lemma}
If $M$ is a $k\times k$ matrix and $N$ is
an $r\times k$ matrix, and the scalars $\lambda_1,\dots,\lambda_t$
are distinct, then
\begin{align}
\sum_{i=1}^t 
\Rank{\TwoByOne{M-\lambda_i I}{N}}
\ge
(t-1)k + \Rank{N}.
\label{eq:300}
\end{align}
\label{lemma:R9.2}
\end{lemma}

We thank Nghi Nguyen for supplying the following clean proof of this result.

\begin{proof}
Let $E_i$ be the null space of $M-\lambda_i I$,
and let $E$ be the null space of $N$.  Then
\begin{align*}
\Rank{\TwoByOne{M-\lambda_i I}{N}} 
= k - \mbox{dim}(E_i \cap E)
\end{align*}
and
\begin{align*}
\Rank{N} = k - \mbox{dim}(E).
\end{align*}
So \eqref{eq:300} is equivalent to
\begin{align*}
tk - \sum_i \mbox{dim} (E_i \cap E) \ge tk - \mbox{dim}(E)
\end{align*}
and hence to
\begin{align*}
\sum_i \mbox{dim} (E_i \cap E) \le \mbox{dim}(E),
\end{align*}
and the latter inequality is true because the subspaces $(E_i \cap E)$
are linearly independent in $E$.  (If $\vv \in E$ is the sum of vectors
$\vv_i \in E_i \cap E$ for $1 \le i \le t$, then we can recover the
vectors $\vv_i$ from $\vv$ using formulas such as
\begin{align*}
(\lambda_1-\lambda_2)\dots(\lambda_1-\lambda_t)\vv_1
= (M-\lambda_2I)\dots(M-\lambda_tI)\vv.)
\end{align*}
\end{proof}

Now, we have 
\begin{align*}
\Rank{\TwoByOne{R_8 R_5 - I}{T_8 R_5}} &\le 2n - k_c - k_b
\end{align*}
from \eqref{eq:R10.1}.  Since
\begin{align*}
\TwoByOne{R_8 R_5}{T_8 R_5} = \TwoByOne{R_8}{T_8} R_5,
\end{align*}
we have
\begin{align*}
\Rank{\TwoByOne{R_8 R_5}{T_8 R_5}} \le \Rank{R_5} \le k_c.
\end{align*}
Now, as stated on page 2756
of~\cite{Dougherty-Freiling-Zeger04-Insufficiency},
we can find a matrix $Q$ such that 
\begin{align}
\Rank{\ThreeByOne{I + R_8 R_5}{T_8 R_5}Q} = k_a
\label{eq:R10.3}
\end{align}
and
\begin{align*}
\Rank{Q} = k_a - \Rank{\TwoByOne{I + R_8 R_5}{T_8 R_5}},
\end{align*}
so
\begin{align*}
\Rank{\TwoByOne{I + R_8 R_5}{T_8 R_5}} = k_a - \Rank{Q}.
\end{align*}
Substituting these facts into \eqref{eq:R9.1} gives
\begin{align}
&2n - k_c - k_b + k_c + k_a - \Rank{Q}\notag\\
&\ \ \ge 2k_a + \Rank{T_8 R_5}.
\label{eq:R10.6}
\end{align}
But \eqref{eq:R10.3} implies that
\begin{align}
\ThreeByOne{I + R_8 R_5}{T_8 R_5}Q R_2 b \yields R_2 b;
\end{align}
combining this with \eqref{eq:R10.7} and \eqref{eq:R10.8} yields
\begin{align*}
T_2 b,\ 
T_3 b,\ 
T_8 R_5 R_2 b,\ 
QR_2 b
\yields
b.
\end{align*}
Using this with the bound on $\Rank{T_8 R_5}$ obtained from \eqref{eq:R10.6},
we get
\begin{align*}
&n - k_a + n - k_c + 2n - k_a - k_b - \Rank{Q} + \Rank{Q}\notag\\
&\ \ \ge k_b,
\end{align*}
which reduces to the desired inequality \eqref{eq:B221}.

\newpage
\section{Non-Fano network}
\label{sec:non-Fano}

\begin{figure}[h]
\begin{center}
\leavevmode
\hbox{\epsfxsize=0.4\textwidth\epsffile{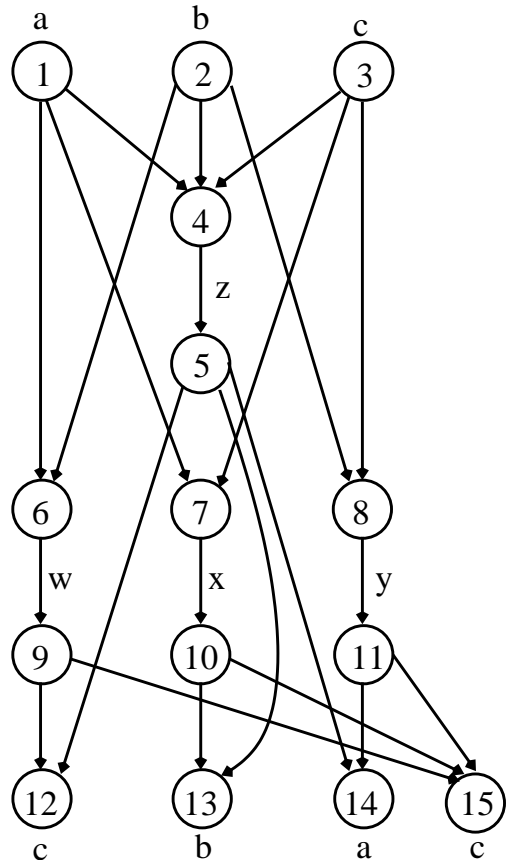}}
\end{center}
\caption{The non-Fano network. 
Source nodes $n_1$, $n_2$, and $n_3$ generate messages $a$, $b$, and $c$, respectively.
Receiver nodes $n_{12}$, $n_{13}$, $n_{14}$, and $n_{15}$ demand messages $c$, $b$, $a$, and $c$, respectively.
The symbol vectors carried on edges $e_{6,9}$, $e_{7,10}$, $e_{8,11}$, $e_{4,5}$ are labeled
as $w$, $x$, $y$, and $z$, respectively.
}
\label{fig:non-Fano-network}
\end{figure}
\begin{theorem}
The achievable rate region
for either linear coding
over any finite field alphabet of odd characteristic
or non-linear coding are the same
for the non-Fano network
and are equal to the closed 
cube
in $\R^3$ whose faces lie on the $6$ planes 
(see Figure~\ref{fig:non-Fano-coding-rate-region}):
\begin{align*}
r_a &= 0\\ 
r_b &= 0\\ 
r_c &= 0\\
r_a &= 1\\ 
r_b &= 1\\
r_c &= 1      
\end{align*}
and whose vertices are the $8$ points:
\begin{align*}
&(0,0,0) &&(0,0,1) &(1,0,0) &&(0,1,0)\ \\
&(0,1,1) &&(1,0,1) &(1,1,0) &&(1,1,1).
\end{align*}
\label{thm:5}
\end{theorem}

\begin{proof}
Consider a network solution over an alphabet $\alphabet$ and
denote the source message dimensions by $k_a$, $k_b$, and $k_c$,
and the edge dimensions by $n$.
Let each source be a random variable whose components are independent and uniformly distributed over $\alphabet$.
Then the solution must satisfy the following inequalities:
\begin{align}
                  k_a &\ge 0                      \label{eq:18}\\
                  k_b &\ge 0                      \label{eq:19}\\
                  k_c &\ge 0                      \label{eq:20}\\
                  k_a &= H(a) = H(z|b,c)   \le H(z) \le n \label{eq:21}\\ 
                  k_b &= H(b) = H(z|a,c)   \le H(z) \le n \label{eq:22}\\ 
                  k_c &= H(c) = H(z|a,b)   \le H(z) \le n. \label{eq:23}
\end{align}
\eqref{eq:18}--\eqref{eq:20} are trivial;
\eqref{eq:21} follows because $z,b,c \yields z,y \yields a$ (at node $n_{14}$), so $z,b,c \yields a,b,c$ and thus $H(a,b,c) = H(z,b,c)$.
\eqref{eq:22} follows because $z,a,c \yields z,x \yields b$ (at node $n_{13}$), so $z,a,c \yields a,b,c$ and thus $H(a,b,c) = H(z,a,c)$.
\eqref{eq:23} follows because $z,a,b \yields z,w \yields c$ (at node $n_{12}$), so $z,a,b \yields a,b,c$ and thus $H(a,b,c) = H(z,a,b)$.
Dividing each inequality in \eqref{eq:18}--\eqref{eq:23} by $n$ gives the $8$ bounding planes
stated in the theorem.

Let $r_a = k_a/n$, $r_b = k_b/n$, and $r_c= k_c/n$,
and let $\polytope$ denote the polyhedron in $\R^3$ consisting of all
$3$-tuples $(r_a, r_b, r_c)$ satisfying \eqref{eq:18}--\eqref{eq:23}.
Then $\polytope$ is simply the unit cube shown in
Figure~\ref{fig:non-Fano-coding-rate-region},
and its extreme points are the $8$ points stated in the theorem.
To show that the $8$ points lie in the achievable rate region,
let $n = k_a = k_b = k_c = 1$
and use the following linear code for $(1,1,1)$ over any odd-characteristic finite field:
\begin{align*}
&w=a+b,\ y=b+c,\ x=a+c, z=a+b+c
\end{align*}
(where node $n_{15}$ can recover its demand via
 $c = (w - y + x) \cdot 2^{-1}$).
The other $7$ points are obtained by setting certain messages to $0$ in the code for $(1,1,1)$.
Since the above codes are all linear, 
the achievable rate regions for linear and non-linear
codes are the same.
\end{proof}

\begin{theorem}
The achievable rate region
for linear coding
over any finite field alphabet of even characteristic
for the non-Fano network
is equal to the closed polyhedron in $\R^3$ whose faces lie on the $7$ planes 
(see Figure~\ref{fig:non-Fano-coding-odd-rate-region}):
\begin{align*}
r_a &= 0\\ 
r_b &= 0\\ 
r_c &= 0\\
r_a &= 1\\
r_b &= 1\\
r_c &= 1\\
r_a + r_b + r_c &= 5/2
\end{align*}
and whose vertices are the $10$ points:
\begin{align*}
&(0,0,0) &&(0,0,1) &(1,0,0) &&(0,1,0)\ \\
&(0,1,1) &&(1,0,1) &(1,1,0)\ \\
&(1,1,1/2) &&(1,1/2,1) &(1/2,1,1).
\end{align*}
\label{thm:8}
\end{theorem}

\begin{proof}
The six inequalities from Theorem~\ref{thm:5} still apply here;
the proof that the
additional inequality
\begin{align}
2k_a + 2k_b + 2k_c &\le 5n \label{eq:23b}
\end{align}
must also hold in the case of even-characteristic finite fields
is given in Section~\ref{sec:NonFanoProof} (and another proof is
given in Section~\ref{sec:NonFanoProof2}).

The new inequality (\ref{eq:23b}) cuts down the achievable rate region
to the polyhedron shown in Figure~\ref{fig:non-Fano-coding-odd-rate-region},
whose extreme points are the 10 points listed in the theorem.
The point $(1,1,1/2)$ is achieved by the following code
with $n = k_a = k_b = 2$ and $k_c = 1$, which works
over any finite field:
\begin{align*}
&w=(a_1,b_1),\ y=(b_1+c,b_2),\ x=(a_1+c,a_2), z=(a_1+b_1+c,a_2+b_2).
\end{align*}
The other two new extreme points are achieved by permuting the variables
in the above code.
\end{proof}

Note that both the
uniform capacity and average capacity are $5/6$
for the non-Fano network, 
for any even-characteristic finite field.

\begin{figure}[h]
\begin{center}
\leavevmode
\hbox{\epsfxsize=0.5\textwidth\epsffile{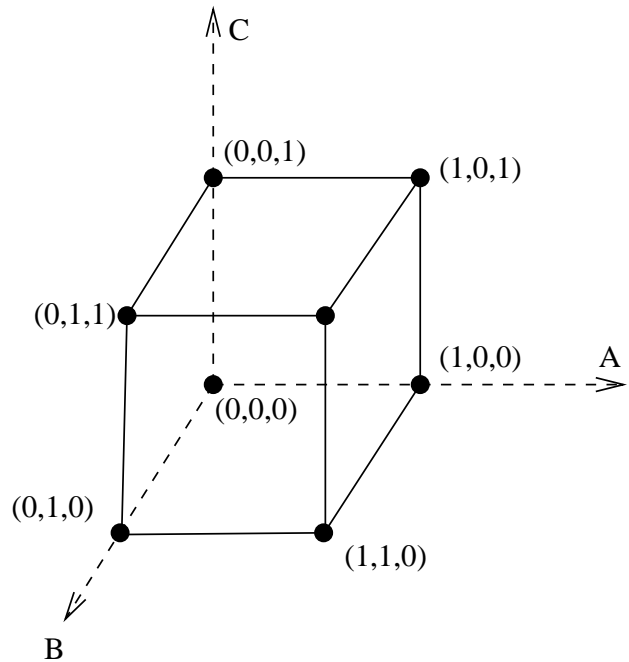}}
\end{center}
\caption{The achievable coding rate region for the Fano network
is a cube in $\R^3$.
}
\label{fig:non-Fano-coding-rate-region}
\end{figure}
\begin{figure}[h]
\begin{center}
\leavevmode
\hbox{\epsfxsize=0.5\textwidth\epsffile{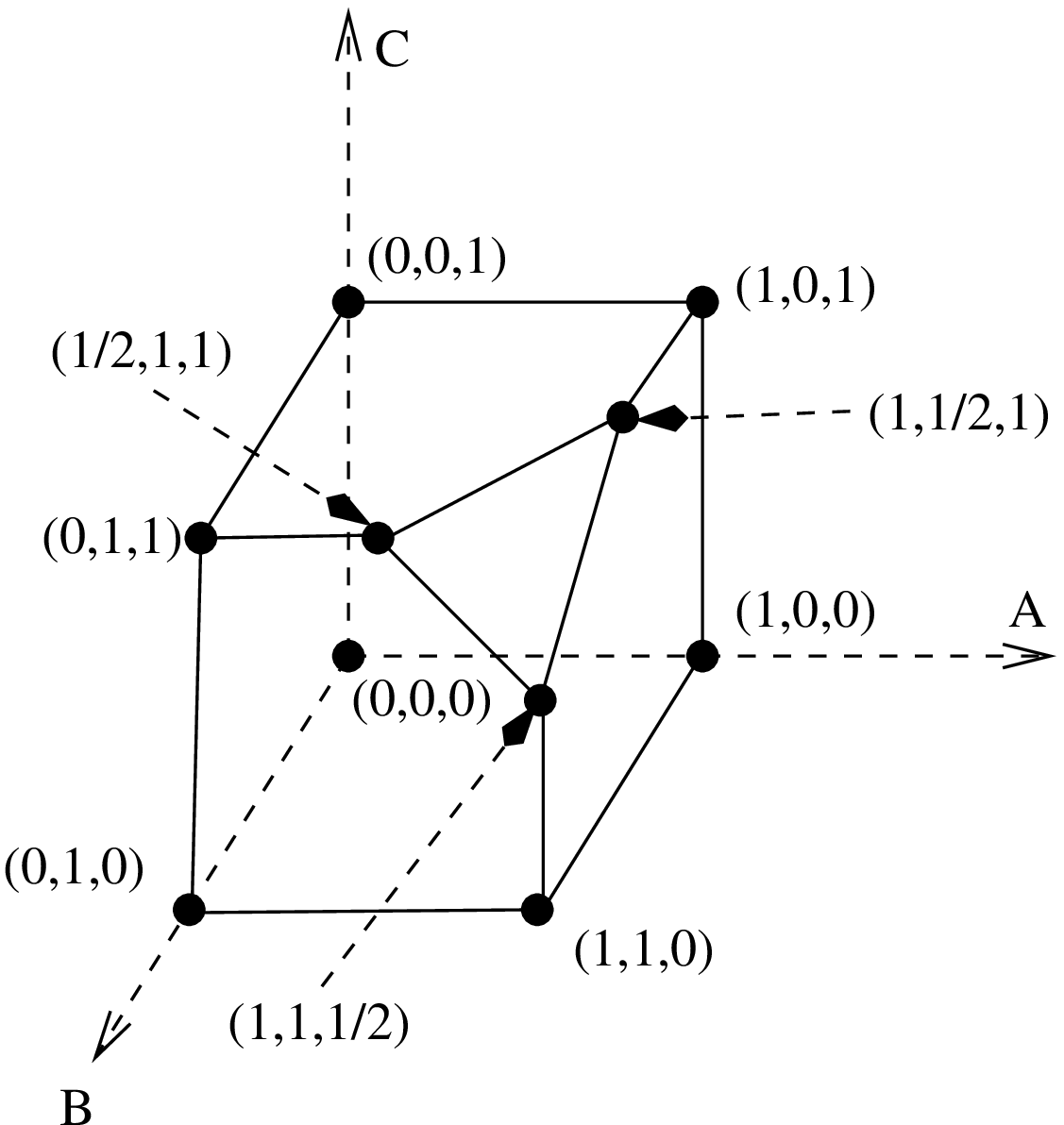}}
\end{center}
\caption{The achievable linear coding rate region over even-characteristic finite fields
for the non-Fano network is a $7$-sided polyhedron with $10$ vertices.
}
\label{fig:non-Fano-coding-odd-rate-region}
\end{figure}
\begin{theorem}
The achievable rate region
for routing for the non-Fano network
is the closed 
tetrahedron
in $\R^3$ whose faces lie on the $4$ planes 
(see Figure~\ref{fig:non-Fano-routing-rate-region}):
\begin{align*}
r_a &= 0\\
r_b &= 0\\
r_c &= 0\\
r_a + r_b + r_c &= 1
\end{align*}
and whose vertices are the $4$ points:
$$(0,0,0),\ (0,0,1),\  (1,0,0),\ (0,1,0).$$
\label{thm:6}
\end{theorem}

\begin{proof}
In addition to satisfying \eqref{eq:18}--\eqref{eq:23},
a routing solution must also satisfy
\begin{align}
k_a + k_b + k_c &\le n \label{eq:24}
\end{align}
since the edge labeled $z$ must carry all $3$ messages $a$, $b$, and $c$.
The inequality \eqref{eq:24} makes the inequalities \eqref{eq:21}--\eqref{eq:23} redundant,
and, in fact, the vertices of the polygon determined by \eqref{eq:18}--\eqref{eq:20} and \eqref{eq:24}
are the $4$ listed in the theorem.
These are achievable using the following routing codes:
\begin{align*}
(0,0,1)&\!:\ \ y=z=c \\
(1,0,0)&\!:\ \ z=a \\
(0,1,0)&\!:\ \ z=b.
\end{align*}
\end{proof}

\begin{figure}[h]
\begin{center}
\leavevmode
\hbox{\epsfxsize=0.5\textwidth\epsffile{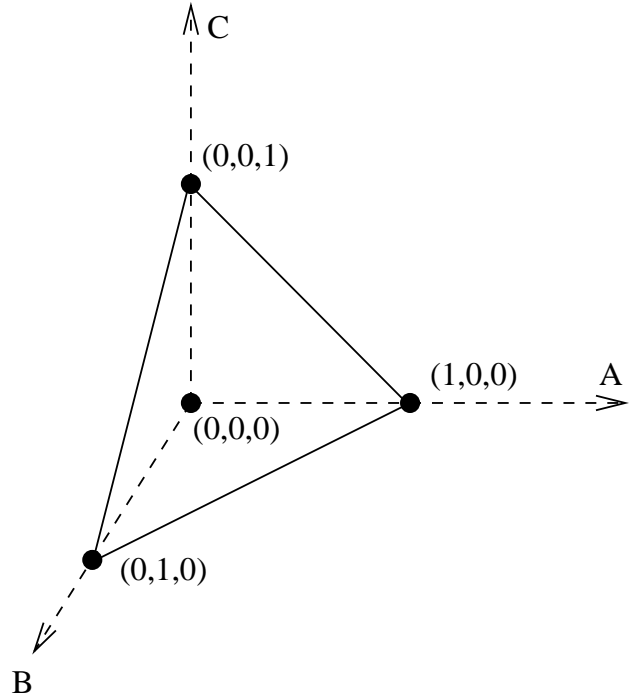}}
\end{center}
\caption{The achievable routing rate region for the Fano network
is a tetrahedron in $\R^3$.
}
\label{fig:non-Fano-routing-rate-region}
\end{figure}

\newpage
\section{Proof of remaining bound for the non-Fano network}
\label{sec:NonFanoProof}

For the case of linear coding over a finite field of characteristic~2,
we want to prove the bound:
\begin{align}
             2k_a + 2k_b + 2k_c \le 5n \label{eq:B222}
\end{align}
We will again do this by following the arguments from
Section~IV of \cite{Dougherty-Freiling-Zeger04-Insufficiency},
with minor modifications.  (Those arguments were for a different
network which was two copies of the non-Fano network with
one demand node merged, but a number of them concentrated on just
the left half of that network and hence will be directly
applicable to the non-Fano network.)

The matrices $M_1$ through $M_{15}$ will be the same as they are
on pages 2756--2757 of \cite{Dougherty-Freiling-Zeger04-Insufficiency};
they label a part of the network there which is identical to the
non-Fano network.  Again here, instead of one value $\delta = n-k$
we have three values $\delta_a = n-k_a$, $\delta_b = n - k_b$,
and $\delta_c = n - k_c$.  When we talk about thinking of an edge
vector as one part of length $k$ followed by one part of length
$n-k$, we will use $k = k_c$ here; so, for instance, $R_7$ is a
$k_c \times k_a$ matrix, while $R_9$ is $k_c \times k_c$.

Now follow the argument from pages 2756--2757 of \cite{Dougherty-Freiling-Zeger04-Insufficiency}
as written, except
that $L$ is just the five vectors
\begin{align*}
&M_3 a + M_4 c , \\
&M_5 b + M_6 c , \\ 
&Q_{13} (M_7 a + M_9 c) , \\
&Q_{15} (M_8 b + M_9 c) , \\
&Q_{10} (M_1 a + M_2 b) 
\end{align*}
without any ``corresponding five objects'' from the other side.
The same argument then yields $L \yields a,b,c$.
Since $M_{15}M_7 = I_{k_a}$, we have $\Rank{M_{15}} \ge k_a$
and hence $\Rank{Q_{15}} \le \delta_a$; similarly,
$\Rank{Q_{13}} \le \delta_b$.  Therefore, following the
computation on page 2757 of \cite{Dougherty-Freiling-Zeger04-Insufficiency}, 
we find that $L$ has only
\begin{align*}
&n + n + [\delta_a + \delta_b - (k_c - \alpha)] + [n - \alpha]\\
&= 2n + \delta_a + \delta_b + \delta_c
\end{align*}
independent entries.  Therefore,
\begin{align*}
2n + \delta_a + \delta_b + \delta_c \ge k_a + k_b + k_c,
\end{align*}
so
\begin{align*}
2k_a + 2k_b + 2k_c \le 5n.
\end{align*}

\newpage
\section{\Vamos{} network}
\label{sec:Vamos}

\begin{figure}[h]
\begin{center}
\leavevmode
\hbox{\epsfxsize=0.450\textwidth\epsffile{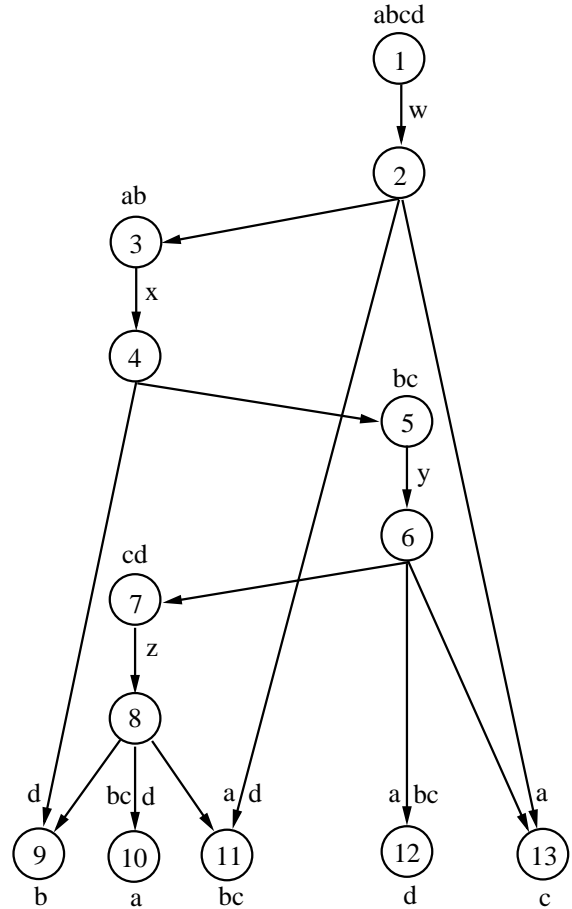}}
\end{center}
\caption{
The \Vamos{} network.
A message variable $a$, $b$, $c$, or $d$ labeled above a node
indicates an in-edge (not shown) from the source node (not shown) generating the  message.
Demand variables are labeled below the receivers $n_9$--$n_{13}$ demanding them.
The edges $e_{1,2}$, $e_{3,4}$, $e_{5,6}$, and $e_{7,8}$
are denoted by
$w$, $x$, $y$, and $z$,
respectively.
}
\label{fig:Vamos-network}
\end{figure}
%

%
\begin{theorem}
The achievable rate region for routing
for the \Vamos{} network
is the polytope in $\R^4$ whose faces lie on the $6$ planes:
\begin{align*}
r_a &= 0\\
r_b &= 0\\
r_c &= 0\\
r_d &= 0\\
2r_a + r_b + 2r_d &= 2\\
 r_a + r_b + r_c + 2r_d &= 2
\end{align*}
and whose vertices are the points
\begin{align*}
&(0, 0, 0, 0) &&(1, 0, 0, 0) &(0, 0, 0, 1)\\     
&(1, 0, 1, 0) &&(0, 2, 0, 0) &(0, 0, 2, 0)
\end{align*}
\label{thm:Vamos-routing}
\end{theorem}

\begin{proof}

The first $4$ planes are trivial.

Now, notice that in a routing solution,
$y$ must carry all of $a$ and $d$ in order to meet
the demands at nodes $n_{10}$ and $n_{12}$, respectively.
Thus, $x$ must carry all of $a$ and $d$ too.
Also, $x$ and $y$ together must carry all of $b$
in order to meet the demand at node $n_9$.
In summary, $x$ and $y$ together must carry at least $2$ copies of $a$,
$2$ copies of $d$, and one copy of $b$.
This implies
$2k_a + k_b + 2k_d \le 2n$,
and therefore $2r_a + r_b + 2r_d \le 2$.

Similarly,
$w$ must carry all of $d$ in order to meet the demand at node $n_{12}$,
and $w$ and $y$ together must carry all of $b$ and $c$ in order to
meet the demands at nodes $n_{11}$ and~$n_{13}$.
Since $y$ must carry all of $a$ and $d$,
we conclude that
$w$ and $y$ together must carry at least 
one copy of $a$,
one copy of $b$,
one copy of $c$,
and
two copies of $d$.
This implies
$k_a + k_b + k_c + 2k_d \le 2n$,
and therefore $r_a + r_b + r_c + 2r_d \le 2$.

It is easy to check that the vertices of the polytope bounded
by the 6 planes listed in the theorem are the 6 vertices
listed in the theorem.
Each of the $6$ vertices can be achieved as follows:
$(0 0 0 0)$ trivially;
$(1 0 0 0)$ with $x=y=z=a$;
$(0 0 0 1)$ with $w = x = y = z = d$;
$(1 0 1 0)$ with $w = c$ and  $x = y = z = a$;
$(0 2 0 0)$ with $w = x = b_1$  and $y = z = b_2$;
$(0 0 2 0)$ with $w = x = c_1$ and $y = z = c_2$.

\end{proof}

%

The following theorem uses only Shannon-type information inequalities
to obtain a polytopal outer bound in $\R^4$ to the achievable rate region.

\begin{theorem}
The achievable rate region
for the \Vamos{} network
lies inside the polytope in $\R^4$ whose faces lie on the $9$ planes:
%
%
\begin{align*}
r_a &= 0        \\
r_b &= 0        \\
r_c &= 0        \\
r_d &= 0        \\
r_a &= 1       \\
r_d &= 1       \\
r_b + r_c &= 2 \\
r_a + r_b &= 2 \\
r_c + r_d &= 2   
\end{align*}
and whose vertices are the points:
\begin{align*}
&(0, 2, 0, 1) && 
(0, 2, 0, 0) && %
(1, 1, 1, 0) && 
(1, 1, 0, 0) \\ %
&(1, 1, 0, 1) && 
(1, 0, 0, 1) && %
(0, 0, 0, 1) && %
(0, 0, 0, 0) \\ %
&(1, 0, 0, 0) && %
(1, 0, 1, 1) && 
(0, 0, 1, 1) && %
(0, 1, 1, 1) \\ 
&(1, 0, 2, 0) && 
(0, 0, 2, 0) && %
(1, 1, 1, 1).    
\end{align*}

\label{thm:Vamos-Shannon}
\end{theorem}

\begin{proof}
Consider a network solution over an alphabet $\alphabet$ and
denote the source message dimensions by $k_a$, $k_b$, $k_c$, and $k_d$,
and the edge dimensions by $n$.
Let each source be a random variable whose components are independent and uniformly distributed over $\alphabet$.
Then the solution must satisfy the following inequalities:
\begin{align}
                  k_a &\ge 0  \label{eq:v1}\\
                  k_b &\ge 0  \label{eq:v2}\\
                  k_c &\ge 0  \label{eq:v3}\\
                  k_d &\ge 0  \label{eq:v4}\\
                  k_a &= H(a) \le H(z|b,c,d)    \le n \label{eq:v5}\\ 
                  k_d &= H(d) \le H(y|a,b,c)    \le n \label{eq:v6}\\ 
            k_b + k_c &= H(b,c) \le H(w,z|a,d)\notag\\
                      &\le H(w,z) \le 2n \label{eq:v7}\\
            k_a + k_b &= H(a,b) \le H(x,z|c,d) \notag\\
                      &\le H(y,z) \le 2n \label{eq:v8}\\
            k_c + k_d &= H(c,d) \le H(w,y|a,b) \notag\\
                      &\le H(w,y) \le 2n. \label{eq:v9}
\end{align}
\eqref{eq:v1}--\eqref{eq:v4} are trivial;
\eqref{eq:v5} follows because $b,c,d,z \yields a$;
\eqref{eq:v6} follows because $a,b,c,y \yields d$;
\eqref{eq:v7} follows because $a,d,w,z \yields b,c$;
\eqref{eq:v8} follows because $x,z,c,d \yields a,b$;
\eqref{eq:v9} follows because $w,y,a,b \yields c,d$;
Dividing each inequality in \eqref{eq:v1}--\eqref{eq:v9} by $n$ gives the $9$ bounding hyperplanes
stated in the theorem.

Let $r_a = k_a/n$, $r_b = k_b/n$, $r_c= k_c/n$, and $r_d = k_d/n$,
and let $\polytope$ denote the polytope in $\R^4$ consisting of all
$4$-tuples $(r_a, r_b, r_c, r_d)$ satisfying \eqref{eq:1}--\eqref{eq:9}.
Then \eqref{eq:v1}--\eqref{eq:v4} and \eqref{eq:v8}--\eqref{eq:v9} 
ensure that $\polytope$ is bounded.
One can easily calculate that each point in $\R^4$ that 
satisfies some independent set of four of the inequalities
\eqref{eq:v1}--\eqref{eq:v9} with equality
and also satisfies the remaining five inequalities
must be one of the $15$ points stated in the theorem.
\end{proof}

For further bounds, we use the following result
from \cite{Dougherty-Freiling-Zeger11-nonShannon}:

Suppose that $A$, $B$, $C$, and $D$ are  random variables and
we have an information inequality of the form
\begin{align}
&a_1 I(A;B) \notag\\
&\ \ \le a_2 I(A;B|C)
+ a_3 I(A;C|B)
+ a_4 I(B;C|A)\notag\\
&\ \ + a_5 I(A;B|D)
+ a_6 I(A;D|B)
+ a_7 I(B;D|A)\notag\\
&\ \ + a_8 I(C;D)
+ a_9 I(C;D|A)
+ a_{10} I(C;D|B).
\label{eq:105}
\end{align}
Then we get the following bound on the \Vamos{} message and edge entropies:
\begin{align}
& (a_2+a_3+a_4) H(a) \notag\\
&\ \ + (a_2+a_3+a_8+a_9+a_{10}) H(b) \notag\\
&\ \ + (a_5+a_7+a_8+a_9+a_{10}) H(c) \notag\\
&\ \ + (a_5+a_6+a_7) H(d) \notag\\
&\ \ + (a_2-a_1-a_7) I(c;y) \notag\\
&\ \ + (a_4+a_7-a_{10}) I(b;x)\notag\\
&\le
(a_5+a_6+a_7+a_8+a_9+a_{10}) H(w) \notag\\
&\ \ + (a_2+a_3+a_4+a_7) H(x) \notag\\
&\ \ + (-a_1+a_2+a_5+a_9) H(y) \notag\\
&\ \ + (a_3+a_8+a_{10}) H(z).
 \label{eq:103}
\end{align}
And by the same argument, if (\ref{eq:105}) is a linear rank inequality
(for a particular characteristic), then (\ref{eq:103}) holds for
any linear (for that characteristic) fractional code for the
\Vamos{} network.

If the inequalities
\begin{align}
&a_2 \ge a_1 + a_7\notag\\
&a_4 + a_7 \ge a_{10}
\label{eq:106}
\end{align}
are satisfied,
then the inequality \eqref{eq:103} directly leads to a
\Vamos{} achievable rate region bound,
by neglecting the (nonnegative) terms involving $I(c;y)$ and $I(b;x)$.
Specifically, in this case,
by substituting
\begin{align*}
H(a)&=k_a\\
H(b)&=k_b\\
H(c)&=k_c\\
H(d)&=k_d\\
H(w)&=H(x)=H(y)=H(z)=n
\end{align*}
into \eqref{eq:103},
we obtain
\begin{align}
&k_a(a_2+a_3+a_4)\notag\\
&\ \
+ k_b(a_2+a_3+a_8+a_9+a_{10})
\notag\\
&\ \
+ k_c(a_5+a_7+a_8+a_9+a_{10})
\notag\\
&\ \
+ k_d(a_5+a_6+a_7)
\notag\\
&\le
n(
-a_1
+ 2a_2
+ 2a_3
+ a_4
+ 2a_5\notag\\
&\ \
+ a_6
+ 2a_7
+ 2a_8
+ 2a_9
+ 2a_{10}
).
 \label{eq:104}
\end{align}
%

%
\begin{theorem}
The achievable rate region
for linear coding
over any finite field alphabet
for the \Vamos{} network
is the polytope in $\R^4$ whose faces lie on the $10$ planes:
\begin{align*}
 r_a &= 0                    \\ 
 r_b &= 0                    \\ 
 r_c &= 0                    \\ 
 r_d &= 0                    \\ 
 r_a &= 1                   \\ 
 r_d &= 1                   \\ 
 r_b + r_c &= 2             \\ 
 r_a + r_b &= 2             \\ 
 r_c + r_d &= 2             \\ 
 r_a + 2r_b + 2r_c +r_d &= 5   
\end{align*}
and whose vertices are the points
\begin{align*}
&(0, 0, 2, 0) &&(0, 0, 1, 1) &&(1, 0, 1, 1) &&(1, 0, 0, 0)\\
&(0, 0, 0, 0) &&(0, 0, 0, 1) &&(1, 0, 0, 1) &&(1, 1, 0, 1)\\
&(1, 1, 0, 0) &&(0, 2, 0, 0) &&(1, 1, 1/2, 1) &&(1, 1/2, 1, 1)\\
%
&(0, 2, 0, 1) &&(1, 1, 1, 0) &&(0, 1, 1, 1) &&(1, 0, 2, 0).
\end{align*}
\label{thm:Vamos-linear}
\end{theorem}

\begin{proof}
The first nine bounding planes come from Theorem~\ref{thm:Vamos-Shannon}.
The tenth bounding plane is shown by letting \eqref{eq:105} be the
Ingleton inequality~\cite{Ingleton}, which can be written in the form
\begin{align*}
I(A;B)\le I(A;B|C) + I(A;B|D) + I(C;D)
\end{align*}
and which is a linear rank inequality for all characteristics,
to get the \Vamos{} linear rate region bound
\begin{align*}
H(a) + 2H(b) + 2H(c) + H(d) \le
2H(w) + H(x) + H(y) + H(z)
\end{align*}
from \eqref{eq:103}.

The proof that the extreme points of the polytope bounded by these planes
are the 16 points listed above is left as an exercise for the reader's
computer (we used {\tt cddlib}~\cite{cddlib}).

Here are linear codes over an arbitrary field) achieving six of the
extreme points:
\begin{align*}
(1,1,1,0)&\!:\ \ n=1, \\
           & w=a+c\\
           & x=a\\
           & y=z=a+b\\
(0,1,1,1)&\!:\ \ n=1, \\
           & w=x=b+d\\
           & y=b+c+d\\ 
           & z=c\\
(1,0,2,0)&\!:\ \ n=1, \\
           & w=c_1\\
           & x=a\\
           & y=z=a+c_2\\
(0,2,0,1)&\!:\ \ n=1, \\
           & w=x=b_1+d\\
           & y=z=b_2+d\\
(1,1,1/2,1)&\!:\ \ n=2, \\
           & w=(b_2+d_1,c+d_2)\\
           & x=(a_1+d_1,a_2+b_2+c+d_2)\\
           & y=(a_1+b_1+d_1,a_2+d_2)\\
           & z=(a_1+b_1,a_2+c)\\
(1,1/2,1,1)&\!:\ \ n=2, \\
           & w=(c_1+d_1,b+d_2)\\
           & x=(a_1+c_1+d_1,a_2+d_2)\\
           & y=(a_1+d_1,a_2+b+c_2+d_2)\\
           & z=(a_1+c_2,a_2+b)
\end{align*}
The remaining $10$ points are achieved by fixing certain messages to be $0$.
\end{proof}


The following theorem uses the non-Shannon-type Zhang-Yeung information inequality
to obtain an additional outer bound in $\R^4$ to the achievable rate region.

\begin{theorem}
The achievable rate region
for non-linear coding
for the \Vamos{} network is bounded by the inequalities:
\begin{align}
4r_a + 4r_b + 2r_c +  r_d &\le 10  \label{eq:vzy1}  \\
2r_a + 2r_b + 4r_c + 4r_d &\le 11  \label{eq:vzy2}  \\
 r_a + 2r_b + 4r_c + 5r_d &\le 11  \label{eq:vzy3}  \\
5r_a + 6r_b + 6r_c + 5r_d &\le 20. \label{eq:vzy4}    
\end{align}
\label{thm:Vamos-Zhang-Yeung}
\end{theorem}

\begin{proof}
If we let \eqref{eq:105} be the Zhang-Yeung inequality~\cite{Zhang-Yeung-July98},
which can be written in the form
\begin{align}
I(A;B)\le 2I(A;B|C) + I(A;C|B) + I(B;C|A) + I(A;B|D) + I(C;D), \label{eq:ZY}
\end{align}
then we get the \Vamos{} network bound
\begin{align}
4H(a) + 4H(b) + 2H(c) + H(d) + I(c;y) \le
2H(w) + 4H(x) + 2H(y) + 2H(z)
\label{eq:avzy1}
\end{align}
from \eqref{eq:103}.
This immediately gives the inequality \eqref{eq:vzy1} (we can simply
discard the $I(c;y)$ term).

Also, we can let \eqref{eq:105} be \eqref{eq:ZY} with variables $C$ and $D$ interchanged; 
then the result from~\eqref{eq:103} is
\begin{align}
H(a) + 2H(b) + 4H(c) + 4H(d) - I(c;y) + I(b;y) \le
5H(w) + 2H(x) + 2H(y) + H(z).
\label{eq:avzy2}
\end{align}
This does not directly give a rate region bound, because the
term $-I(c;y)$ cannot be simply discarded.  However, if we
add \eqref{eq:avzy1} and \eqref{eq:avzy2}, we get an inequality that
yields \eqref{eq:vzy4}; if we add to \eqref{eq:avzy2} the inequality
$H(a) + I(c;y) \le H(y)$ (which, as noted in
\cite{Dougherty-Freiling-Zeger11-nonShannon}, holds in the
\Vamos{} network because $b,c,d,y \yields a$), we get
\eqref{eq:vzy2}; and if we add to \eqref{eq:avzy2} the inequality
$H(d) + I(c;y) \le H(y)$ (which, as noted in
\cite{Dougherty-Freiling-Zeger11-nonShannon}, holds in the
\Vamos{} network because $a,b,c,y \yields d$), we get
\eqref{eq:vzy3}.
\end{proof}


Many additional non-Shannon-type information inequalities
are given in~\cite{Dougherty-Freiling-Zeger11-nonShannon}.
These can be used as above to give additional bounds on the
achievable rate region
for non-linear coding
for the \Vamos{} network.
In fact, the inequalities from~\cite{Dougherty-Freiling-Zeger11-nonShannon}
using at most four copy variables with at most three copy steps
yield 158 independent constraints on this achievable rate region.
(Note: inequalities \eqref{eq:vzy1}--\eqref{eq:vzy4} are superseded
by these new inequalities.)  One of these is used
in~\cite{Dougherty-Freiling-Zeger11-nonShannon} to show that the
uniform coding capacity of the \Vamos{} network is at most $19/21$.

Since there are infinitely many information inequalities on four
random variables~\cite{Matus}, it is quite possible that the
achievable rate region
for non-linear coding
for the \Vamos{} network is not a polytope.  On the other hand,
this rate region could be quite simple; to date, no fractional solution
is known for the \Vamos{} network which lies outside the achievable
rate region for linear coding.

\newpage
\section{New Linear Rank Inequalities from Networks}

We now give a new method for producing bounds on achievable rate regions
for linear coding.  Unlike the previous method using matrix algebra,
this method actually produces explicit linear rank inequalities
(perhaps only true for some characteristics) which directly imply
the bounds in question.  However, it is not clear yet that this new
method can produce all results obtained from the matrix algebra method.

In particular, we produce an explicit linear rank inequality valid
only for odd-characteristic fields, and another linear rank inequality valid
only for even-characteristic fields.  Such inequalities have also been
produced by Blasiak, Kleinberg, and Lubetzky~\cite{Blasiak-Kleinberg-Lubetzky11}
(also by use of the Fano and non-Fano matroids), but those inequalities
do not directly give bounds for the networks here.

We start by giving some basic results in linear algebra.

If $A$ is a subspace of a finite-dimensional vector space $V$,
then we denote the codimension of $A$ in $V$ by
$\Codim{V}{A} = \Dim{V} - \Dim{A}.$ 

\begin{lemma}
For any subspaces $A_1, \dots, A_m$ of finite-dimensional vector space $V$,
$$\Codim{V}{\displaystyle\bigcap_{i=1}^m A_i} \le \sum_{i=1}^m \Codim{V}{A_i}.$$
\label{lem:2}
\end{lemma}

\begin{lemma}
Let $A$ and $B$ be finite-dimensional vector spaces,
let $f:A\to B$ be a linear function,
and let $B'$ be a subspace of $B$.
Then
$\Codim{A}{f^{-1}(B')} \le \Codim{B}{B'}$.
\label{lem:3}
\end{lemma}

\begin{proof}
Let $S = f^{-1}(B')$
and let $T$ be a subspace of $A$ such that $S+T=A$
and $S\cap T = \{0\}$.
Let $g:T\to B$ be a linear function such that $g=f$ on $T$.
Then we have
\begin{align*}
\Codim{A}{S}
&= \Dim{T} & \Comment{$S+T=A$ and $S\cap T=\{0\}$}\\
&= \Dim{g(T)} + \Nullity{g} \\
&= \Dim{g(T)} & \Comment{$g^{-1}(\{0\}) = \{0\}$}\\
&\le \Codim{B}{B'}. & \Comment{$B' \cap g(T)  = \{0\}$}
\end{align*}
\end{proof}

\begin{lemma}
Let $A_1, \dots, A_k, B$ be subspaces of a finite-dimensional vector space $V$.
There exist linear functions 
$f_i:B \to A_i$ 
(for $i=1, \dots, k$)
such that 
$f_1 + \dots + f_k = I$ 
on a subspace of $B$ 
of codimension $H(B|A_1, \dots, A_k)$ in $B$.
\label{lem:1}
\end{lemma}

\begin{proof}
The subspace is $W=(A_1 + \dots + A_k) \cap B$.
For each $w_j$ in a basis for $W$,
choose $x_{i,j}\in A_i$ for $i=1, \dots, k$ 
such that $w_j=x_{1,j}+\dots+x_{k,j}$.
Define linear maps $g_i:W \to A_i$
for $i=1, \dots, k$ so that $g_i(w_j) = x_{i,j}$
for all $i$ and $j$; then extend each $g_i$ arbitrarily 
to a linear map $f_i:B \to A_i$.
We have
$H(B|A_1, \dots, A_k) = \Dim{B} - \Dim{B \cap (A_1 + \dots + A_k)}
                      = \Dim{B} - \Dim{W}$.
\end{proof}

\begin{lemma}
Let $A,B,C$ be subspaces of a finite-dimensional vector space $V$,
and let
$f:A\to B$ and $g:A\to C$ be linear functions 
such that $f+g=0$ on $A$.
Then $f=g=0$ on a subspace of $A$ of codimension at most $I(B;C)$ in $A$.
\label{lem:2a}
\end{lemma}

\begin{proof}
For all $u\in A$, $g(u)\in B$ so $f(u) = -g(u)\in B$
and therefore $f$ maps $A$ into $B\cap C$.
Thus,
$\Dim{A} - \Nullity{f} = \Rank{f} \le \Dim{B\cap C} = I(B;C)$,
so the kernel of $f$ has codimension at most $I(B;C)$ in $A$.
\end{proof}

\begin{lemma}
Let $A,B_1,\dots, B_k$ be subspaces of a finite-dimensional vector space $V$,
and let
$f_i:A\to B_i$ be linear functions 
such that $f_1+\dots+f_k=0$ on $A$.
Then $f_1=\dots=f_k=0$ on a subspace of $A$ of codimension at most 
$H(B_1)+ \dots + H(B_k) - H(B_1,\dots,B_k)$
in $A$.
\label{lem:2b}
\end{lemma}

\begin{proof}
%
Use induction on $k$.
The claim is trivially true for $k=1$,
and is true for $k=2$ by Lemma~\ref{lem:2a}.
Let us assume it is true up to $k-1$ for $k\ge 3$.
Apply Lemma~\ref{lem:2a} with 
$B=B_k$, 
$C=B_1+\dots+B_{k-1}$,
$f=f_k$,
and
$g=f_1+\dots+f_{k-1}$
to get
$f_1+\dots+f_{k-1}=f_k=0$ on a subspace $S$ of $A$ satisfying
$$\Codim{A}{S} \le H(B_1,\dots,B_{k-1}) + H(B_k) - H(B_1,\dots,B_k).$$
By the induction hypothesis,
$f_1=\dots=f_{k-1}=0$ on a subspace $S'$ of $S$ satisfying
$$\Codim{S}{S'} \le H(B_1) + \dots + H(B_{k-1}) - H(B_1,\dots,B_{k-1}).$$
Adding these two inequalities gives us the desired result for subspace $S'$.
\end{proof}

\newpage
\subsection{A Linear Rank Inequality from the Fano Network}
\label{sec:FanoProof2}

\begin{theorem}
Let $A,B,C,D,W,X,Y,Z$ be subspaces of a finite-dimensional vector space $V$
over a scalar field of odd characteristic.
Then, the following linear rank inequality holds:
\begin{align}
&2H(A) + H(B) + 2H(C) \notag \\
&\ \ \le H(W) + H(X) + H(Y) + H(Z) \notag \\
&\ \ \  + 2H(A|Z,Y) + H(B|X,Z) + 2H(C|A,X) \notag \\
&\ \ \  + 3H(X|W,Y) + 3H(Z|W,C) \notag \\
&\ \ \  + 5H(W|A,B) + 5H(Y|B,C) \notag \\
&\ \ \  + 5( H(A) + H(B) + H(C) - H(A,B,C) ).
\label{oddLRI}
\end{align}
\label{thm:LinearRankInequality1}
\end{theorem}

\begin{proof}
We will use the Fano network in Figure~\ref{fig:Fano-network},
derived in~\cite{Dougherty-Freiling-Zeger07-NetworksMatroidsNonShannon},
from the Fano matroid,
to help guide the proof.
By Lemma~\ref{lem:1},
there exist linear functions
\begin{align*}
f_1:    W &\to A & f_2:    W &\to B\\
f_3:    Y &\to B & f_4:    Y &\to C\\
f_5:    X &\to W & f_6:    X &\to Y\\
f_7:    Z &\to W & f_8:    Z &\to C\\
f_9:    C &\to A & f_{10}: C &\to X\\
f_{11}: B &\to X & f_{12}: B &\to Z\\
f_{13}: A &\to Z & f_{14}: A &\to Y
\end{align*}
such that
\begin{align}
f_1    + f_2    &= I \mbox{ on a subspace $W'$ of $W$ with } \Codim{W}{W'} \le H(W|A,B)\label{eq:41}\\
f_3    + f_4    &= I \mbox{ on a subspace $Y'$ of $Y$ with } \Codim{Y}{Y'} \le H(Y|B,C)\label{eq:42}\\
f_5    + f_6    &= I \mbox{ on a subspace $X'$ of $X$ with } \Codim{X}{X'} \le H(X|W,Y)\notag\\
f_7    + f_8    &= I \mbox{ on a subspace $Z'$ of $Z$ with } \Codim{Z}{Z'} \le H(Z|W,C)\label{eq:44}\\
f_9    + f_{10} &= I \mbox{ on a subspace $C'$ of $C$ with } \Codim{C}{C'} \le H(C|A,X)\notag\\
f_{11} + f_{12} &= I \mbox{ on a subspace $B'$ of $B$ with } \Codim{B}{B'} \le H(B|X,Z)\notag\\
f_{13} + f_{14} &= I \mbox{ on a subspace $A'$ of $A$ with } \Codim{A}{A'} \le H(A|Z,Y).\label{eq:47}
\end{align}
Combining these, we get maps
\begin{align}
f_1 f_7 f_{13}: A \to A\label{eq:38}\\
f_2 f_7 f_{13} + f_3 f_{14}: A \to B\label{eq:39}\\
f_8 f_{13} + f_4 f_{14}: A \to C.\label{eq:40}
\end{align}
Note that
\begin{align*}
f_1 f_7 f_{13} + f_2 f_7 f_{13} &= f_7 f_{13} 
\mbox{ on the subspace } f_{13}^{-1} f_7^{-1} (W') \mbox{ of } A\\
f_7 f_{13} + f_8 f_{13} &= f_{13} 
\mbox{ on the subspace } f_{13}^{-1} (Z') \mbox{ of } A\\
f_3 f_{14} + f_4 f_{14} &= f_{14} 
\mbox{ on the subspace } f_{14}^{-1} (Y') \mbox{ of } A
\end{align*}
so the sum of the functions in \eqref{eq:38}--\eqref{eq:40} 
is equal to $I$ on the subspace 
$$
A'' \doteq
A' \cap
f_{13}^{-1} (Z') \cap
f_{13}^{-1} f_7^{-1} (W') \cap 
f_{14}^{-1} (Y') 
$$
and we get
\begin{align*}
&\Codim{A}{A''}\\
&\le
\Codim{A}{A'}
+ \Codim{A}{f_{13}^{-1} (Z')}\notag\\
&\ \ \ \ + \Codim{A}{f_{13}^{-1} f_7^{-1} (W')}
+ \Codim{A}{f_{14}^{-1} (Y')} &\Comment{Lemma~\ref{lem:2}}\\
&\le
\Codim{A}{A'} + \Codim{Z}{Z'} + \Codim{W}{W'} + \Codim{Y}{Y'} &\Comment{Lemma~\ref{lem:3}}\\
&\le
H(A|Z,Y) + H(Z|W,C) + H(W|A,B) + H(Y|B,C). &\Comment{\eqref{eq:41}, \eqref{eq:42}, \eqref{eq:44},\eqref{eq:47}}
\end{align*}
Applying Lemma~\ref{lem:2b} to
$f_1 f_7 f_{13} - I$,
$f_2 f_7 f_{13} + f_3 f_{14}$,
and
$f_8 f_{13} + f_4 f_{14}$,
we get a subspace $\bar{A}$ of $A''$ such that
\begin{align}
\Codim{A}{\bar{A}} 
&= \Codim{A}{A''} + \Codim{A''}{\bar{A}} \notag\\
&\le \Delta_A \label{eq:30}\\
&\doteq H(A|Z,Y) + H(Z|W,C) + H(W|A,B) + H(Y|B,C)\notag\\
&\ \ \ + H(A) + H(B) + H(C) - H(A,B,C)\label{eq:31}
\end{align}
on which
\begin{align}
f_1 f_7 f_{13} &= I\label{eq:37}\\
f_2 f_7 f_{13} + f_3 f_{14} &=0\notag\\
f_8 f_{13} + f_4 f_{14} &= 0.\notag
\end{align}
Similarly,
we get a subspace $\bar{C}$ of $C$ such that
\begin{align}
\Codim{C}{\bar{C}} &\le \Delta_C\label{eq:32}\\
&\doteq H(C|A,X) + H(X|W,Y) + H(W|A,B) + H(Y|B,C)\notag\\
&\ \ \ + H(A) + H(B) + H(C) - H(A,B,C)\label{eq:33}
\end{align}
on which
\begin{align}
f_4 f_6 f_{10} &= I\label{eq:36}\\
f_2 f_5 f_{10} + f_3 f_6 f_{10} &=0\notag\\
f_9 + f_1 f_5 f_{10} &= 0\notag
\end{align}
and a subspace $\bar{B}$ of $B$ such that
\begin{align}
\Codim{B}{\bar{B}} &\le \Delta_B\label{eq:34}\\
&\doteq H(B|X,Z) + H(X|W,Y) + H(Z|W,C) + H(W|A,B)\notag\\
&\ \ \ + H(Y|B,C) + H(A) + H(B) + H(C) - H(A,B,C)\label{eq:35}
\end{align}
on which
\begin{align*}
f_2 f_5 f_{11} + f_2 f_7 f_{12} + f_3 f_6 f_{11} &= I\\
f_1 f_5 f_{11} + f_1 f_7 f_{12} &=0\\
f_4 f_6 f_{11} + f_8 + f_{12} &= 0.
\end{align*}
Note: There is only one $H(W|A,B)$ in \eqref{eq:35}
because we can write
$$f_i f_5 f_{11} + f_i f_7 f_{12} = f_i (f_5 f_{11} + f_7 f_{12})$$
for $i=1,2$.

Let us define the following subspaces of $B$:
\begin{align}
S_1 &= \{ u\in B: f_{11}u \in f_{10} \bar{C} \}\notag\\
S_2 &= \{ u\in B: f_{12}u \in f_{13} \bar{A} \}\notag\\
S_3 &= \{ u\in B: f_5 f_{11}u \in f_7 f_{13}\bar{A} \}\notag\\
S_4 &= \{ u\in B: f_{14} f_1 f_7 f_{12}u \in f_6 f_{10}\bar{C} \}\notag\\
S &= \bar{B} \cap S_1 \cap S_2 \cap S_3 \cap S_4.\label{eq:S1}
\end{align}
Then we have the following:
\begin{align}
\Codim{B}{S_1} 
&\le \Codim{X}{f_{10}\bar{C}} &\Comment{Lemma~\ref{lem:3}}\notag\\
&= \Dim{X} - \Dim{\bar{C}} &\Comment{\eqref{eq:36} $\yields f_{10}$ injective }\notag\\
&= \Codim{C}{\bar{C}} + H(X) - H(C)\notag\\
&\le \Delta_C + H(X) - H(C) &\Comment{\eqref{eq:32}}\label{eq:26}\\
\Codim{B}{S_2} 
&\le \Codim{Z}{f_{13}\bar{A}}  &\Comment{Lemma~\ref{lem:3}}\notag\\
&= \Dim{Z} - \Dim{\bar{A}} &\Comment{\eqref{eq:37} $\yields f_{13}$ injective }\notag\\
&= \Codim{A}{\bar{A}} + H(Z) - H(A)\notag\\
&\le \Delta_A + H(Z) - H(A)&\Comment{\eqref{eq:30}}\label{eq:27}\\
\Codim{B}{S_3} 
&\le \Codim{W}{f_7 f_{13}\bar{A}} &\Comment{Lemma~\ref{lem:3}}\notag\\
&= \Dim{W} - \Dim{\bar{A}} &\Comment{\eqref{eq:37} $\yields f_7, f_{13}$ injective }\notag\\
&= \Codim{A}{\bar{A}} + H(W) - H(A)\notag\\
&\le \Delta_A + H(W) - H(A)&\Comment{\eqref{eq:30}}\label{eq:28}\\
\Codim{Y}{S_4} 
&\le \Codim{Y}{f_6 f_{10}\bar{A}} &\Comment{Lemma~\ref{lem:3}}\notag\\
&= \Dim{Y} - \Dim{\bar{C}} &\Comment{\eqref{eq:36} $\yields f_6, f_{10}$ injective }\notag\\
&= \Codim{C}{\bar{C}} + H(Y) - H(C)\notag\\
&\le \Delta_C + H(Y) - H(C). &\Comment{\eqref{eq:32}}\label{eq:29}
\end{align}
Suppose $t\in S$.
Then,
\begin{align}
f_2 f_5 f_{11} t + f_2 f_7 f_{12} t
&= f_2 f_7 f_{13} f_1 f_5 f_{11} t + f_2 f_7 f_{12} t\notag\\
&\ \ [ \mbox{ we have } f_5 f_{11} t = f_7 f_{13} u 
       \mbox{ for some } u \in \bar{A}, \notag\\
&\ \ \ \mbox{ and } f_7 f_{13} f_1 f_7 f_{13} u = f_7 f_{13} u
       \mbox{ since } f_1 f_7 f_{13} u = u ]\notag\\
&= f_2 f_7 f_{13} f_1 f_5 f_{11} t + f_2 f_7 f_{13} f_1 f_7 f_{12} t\notag\\
&\ \ [ \mbox{ since } f_{12} t \in f_{13} \bar{A} ]\notag\\
&= f_2 f_7 f_{13} ( f_1 f_5 f_{11}  + f_1 f_7 f_{12} ) t\notag\\
&= 0 \label{eq:48}\\
&\ \ [ \mbox{ since } t \in \bar{B} ]\notag
\end{align}
\begin{align}
f_2 f_5 f_{11} t + f_3 f_6 f_{11} t
&= f_2 f_5 f_{10} f_4 f_6 f_{11} t + f_3 f_6 f_{10} f_4 f_6 f_{11} t\notag\\
&\ \ [ \mbox{ since } f_{11} t \in f_{10} \bar{C} ]\notag\\
&= (f_2 f_5 f_{10} t + f_3 f_6 f_{10} ) f_4 f_6 f_{11} t\notag\\
&= 0\label{eq:49}\\
&\ \ [ \mbox{ since } f_{11} t \in f_{10} \bar{C} \mbox{ and hence }\notag\\
&\ \ \ f_4 f_6 f_{11} t \in f_4 f_6 f_{10} \bar{C} = \bar{C} ]\notag
\end{align}
\begin{align}
f_2 f_7 f_{12} t + f_3 f_6 f_{11} t
&= f_2 f_7 f_{12}t + f_3 f_6 f_{10} f_4 f_6 f_{11} t\notag\\
&= f_2 f_7 f_{12}t - f_3 f_6 f_{10} f_8 f_{12} t\notag\\
&= f_2 f_7 f_{12}t - f_3 f_6 f_{10} f_8 f_{13} f_1 f_7 f_{12} t\notag\\
&= f_2 f_7 f_{12}t + f_3 f_6 f_{10} f_4 f_{14} f_1 f_7 f_{12} t\notag\\
&= f_2 f_7 f_{12}t + f_3 f_{14} f_1 f_7 f_{12} t\notag\\
&= f_2 f_7 f_{13} f_1 f_7 f_{12}t + f_3 f_{14} f_1 f_7 f_{12} t\notag\\
&= (f_2 f_7 f_{13} + f_3 f_{14}) f_1 f_7 f_{12} t\notag\\
&= 0.\label{eq:50}
\end{align}
We therefore obtain
\begin{align*}
2t 
&= 2( f_2 f_5 f_{11} t + f_2 f_7 f_{12} t + f_3 f_6 f_{11} t )\\
&= ( f_2 f_5 f_{11} t + f_2 f_7 f_{12} t ) 
 + ( f_2 f_5 f_{11} t + f_3 f_6 f_{11} t )
 + ( f_2 f_7 f_{12} t + f_3 f_6 f_{11} t )\\
&= 0 + 0 + 0 = 0. & \Comment{\eqref{eq:48},\eqref{eq:49},\eqref{eq:50}}
\end{align*}
Since the field has odd characteristic,
we must have $t=0$.
Thus, $S = \{0\}$,
and therefore
\begin{align*}
H(B) &= \Codim{B}{S}\\
&\le \Codim{B}{\bar{B}} + \sum_{i=1}^4 \Codim{B}{S_i}&\Comment{\eqref{eq:S1}, Lemma~\ref{lem:2}}\\
&\le \Delta_B + 2 \Delta_A + 2\Delta_C\\
&\ \ \  + H(W) + H(X) + H(Y) + H(Z) - 2H(A) - 2H(C). &\Comment{\eqref{eq:34},\eqref{eq:26}--\eqref{eq:29}}
\end{align*}
The result then follows from
\eqref{eq:31},
\eqref{eq:33},
and
\eqref{eq:35}.
\end{proof}

In the context of the Fano network, all of the compound terms at
the end of inequality~\eqref{oddLRI} are zero, so this inequality
directly implies inequality~\eqref{eq:02}.

By replacing $W$ with $W\cap (A+B+C+X+Y+Z)$
and similarly for $X$, $Y$, and $Z$,
one can improve the inequality to a balanced form where 
$H(W)$ becomes $I(W; A,B,C,X,Y,Z)$,
$H(W|A,B)$ becomes $I(W;C,X,Y,Z|A,B)$,
and similarly for $X$, $Y$, and $Z$.

\begin{theorem}
The linear rank inequality in Theorem~\ref{thm:LinearRankInequality1}
holds for any scalar field if $\Dim{V} \le 2$,
but may not hold if the scalar field has characteristic $2$
and $\Dim{V} \ge 3$.

\end{theorem}

\begin{proof}
In $V=GF(2)^3$, define the following subspaces of $V$:
\begin{align*}
A &= \Span{(1,0,0)}\\
B &= \Span{(0,1,0)}\\
C &= \Span{(0,0,1)}\\
W &= \Span{(1,1,0)}\\
X &= \Span{(1,0,1)}\\
Y &= \Span{(0,1,1)}\\
Z &= \Span{(1,1,1)}
\end{align*}
It is easily verified that the inequality in Theorem~\ref{thm:LinearRankInequality1}
is not satisfied in this case.

Next we show the inequality indeed holds if $\Dim{V} \le 2$.
One way to do this is to show (using software such as {\tt Xitip}
\cite{Xitip}) that the inequality becomes a Shannon inequality
under the assumption that $H(A) = 0$, or under the assumption
$H(B|A) = 0$, or under the assumption $H(C|A,B) = 0$.  If all three
of these assumptions fail, then we must have
\begin{align}
\Dim{V} \ge H(A,B,C) > H(A,B) > H(A) > 0
\end{align}
and hence $\Dim{V} \ge 3$.

Or one can give a direct argument by cases.
Assume to the contrary that there exist subspaces $A,B,C,W,X,Y,Z$  
of vector space $V$ such that

\begin{align}
&2H(A) + H(B) + 2H(C)\notag\\
&\ \ > H(W) + H(X) + H(Y) + H(Z)\notag\\
&\ \ \  + 2H(A|Z,Y) + H(B|X,Z) + 2H(C|A,X)\notag\\
&\ \ \  + 3H(X|W,Y) + 3H(Z|W,C)\notag\\
&\ \ \  + 5H(W|A,B) + 5H(Y|B,C)\notag\\
&\ \ \  + 5( H(A) + H(B) + H(C) - H(A,B,C) ).
\label{eq:contradiction}
\end{align}
Let $Q = ( H(A), H(B), H(C), H(A,B,C) )$
and
$R = H(A) + H(B) + H(C) - H(A,B,C)$.
Let $\LHS$ and $\RHS$ denote the left and right sides of inequality \eqref{eq:contradiction}.
We will obtain contradictions for all the possible values of $Q$.

\textbf{Case (i)}: $\Dim{V}=1$\\

All entropies are $0$ or $1$. 
Since $\LHS \le 5$,
at most one of $H(A), H(B), H(C)$
can equal $1$, for otherwise $R \ge 1$ would imply $\RHS \ge 5$.
\begin{itemize}
\item [(1001):] 
$\LHS=2$ implies $H(A|Z,Y)=0$ which implies $H(Z)=1$ or $H(Y)=1$. 
Also, we must have $H(Z|W,C)=H(Y|B,C)=0$, the latter implying
$H(Y)=0$.
So we must have $H(Z)=1$ which in turn implies $H(W)=1$ and therefore
$\RHS \ge 2$.

\item [(0101):] 
$\LHS=1$ implies 
$H(B|X,Z)=0$ which implies
$H(X)=1$ or $H(Z)=1$, and
therefore $\RHS \ge 1$.

\item [(0011):] 
$LHS=2$ implies 
$H(C|A,X)=0$ and $H(X|W,Y)=0$,
which imply $H(X)=1$,
which implies $H(W)=1$ or $H(Y)=1$
and therefore $\RHS \ge 2$.
\end{itemize}

\textbf{Case (ii)}: $\Dim{V}=2$\\

All entropies are $0$, $1$, or $2$.
$\LHS \le 10$ implies $\RHS \le 9$, and therefore $R \le 1$.
$\LHS \ge 1$ implies $H(A,B,C) > 0$ and therefore $H(A,B,C) \in \{1,2\}$.

\begin{itemize}
\item [(1011):] $\LHS \le 4$ and $R=1$ imply $\RHS \ge 5$.
\item [(1101):] Same.
\item [(0111):] Same.
\item [(2001):] Same.
\item [(0201):] Same.
\item [(0021):] Same.

\item [(2012):] $\LHS=6$. $R=1$ implies $\RHS \ge 5$
which implies $H(A|Z,Y)=0$ which implies $H(Z,Y) \ge 1$ and therefore $\RHS \ge 6$.
\item [(1022):] Same.

\item [(1112):] $\LHS=5$. $R=1$ implies $\RHS \ge 5$.
\item [(0122):] Same.
\item [(2102):] Same.

\item [(0212):] $\LHS=4$. $R=1$ implies $\RHS \ge 5$.
\item [(1202):] Same.

\item [(1001):] $\LHS=2$ implies $H(A|Z,Y)=0$ which implies $H(Z)=1$ or $H(Y)=1$.
          If $H(Z)=1$, then $H(Z|W,C)=0$ which would imply $H(W)=1$ 
              and therefore $\RHS \ge 2$.
          If $H(Y)=1$, then $H(Z|W,C)=1$ which would imply $\RHS \ge 5$.

\item [(0101):] $\LHS=1$ implies $H(X)=H(Z)=0$ which implies $H(B|X,Z)=1$ 
              and therefore $\RHS \ge 1$.

\item [(0011):] $\LHS=2$ implies $H(C|A,X)=0$ which implies $H(X)=1$.
              Also,  $H(X|W,Y)=0$ implies $H(W,Y) \ge 1$ and therefore $\RHS \ge 2$.

\item [(0202):] $LHS=2$ implies $H(X)+H(Z) \le 1$ which implies $H(B|X,Z) \ge 1$
              which implies $H(B|X,Z)=1$ which implies $H(X,Z)=1$ which implies
               $H(X)+H(Z)=1$ and therefore $\RHS \ge 2$.

\item [(0022):] $\LHS=4$ implies $H(W|A,B)=0$ which implies $H(W)=0$.
              Also, $H(C|A,X) \le 1$ implies $H(X) \ge 1$ which implies
             $H(X|W,Y)=0$ which implies $H(Y) \ge H(X)$.
             Thus, $H(C|A,X)=0$ which implies $X=C$ which implies $H(Y) \ge H(C) = 2$
             and therefore $\RHS \ge 4$.

\item [(2002):] $\LHS=4$ implies $H(Y|B,C)=0$ which implies $H(Y)=0$.
              Also, $H(A|Z,Y) \le 1$ which implies $H(Z) \ge 1$.
              Additionally, $H(Z|W,C)=0$ which implies $H(W) \ge H(Z)$
              which implies $H(A|Z,Y)=0$ which implies $H(Z)=2$
              and therefore $\RHS \ge 4$.

\item [(1102):] $H(A,B,C)=2$ implies that $A \ne B$. 
        $\LHS=3$ implies $H(A|Z,Y)=0$ or $H(B|X,Z)=0$.
        If $H(B|X,Z)=0$, then $H(X)+H(Z) \ge 1$ which implies $\RHS \ge 1$
        and therefore $H(A|Z,Y)=0$.
        So it suffices to assume $H(A|Z,Y)=0$. 
        We have $H(Y|B,C)=0$ which implies $Y$ is a subspace of $B$, 
        which implies $H(Z) \ge 1$. Thus, $H(Z|W,C)=0$ which implies $H(W) \ge 1$,
        so $\RHS \ge 2$. 
        Hence, $H(B|X,Z)=0$ and $H(X)=0$ which imply $Z=B$ and therefore $H(A|Z,Y) \ne 0$.

\item [(0112):] $H(A,B,C)=2$ implies $B \ne C$. 
        $\LHS=3$ implies $H(B|X,Z)=0$ or $H(C|A,X)=0$.
        If $H(B|X,Z)=0$, then $H(X)+H(Z) \ge 1$ which implies $\RHS \ge 1$
        and therefore $H(C|A,X)=0$.
        So it suffices to assume $H(C|A,X)=0$.
        Thus we have $H(X) \ge 1$.
        Also, $H(X|W,Y)=0$ which implies $H(W)+H(Y) \ge H(X)$ and so $\RHS \ge 2$.
        Thus, $H(X)=1$ which implies $X=C$, and therefore $H(W)=1$ or $H(Y)=1$. 
        Since $H(W|A,B)=0$, $W$ is a subspace of $B$ and therefore $Y=C$.
        Finally, $H(B|X,Z)=0$ which implies $H(Z) \ge 1$ and therefore $\RHS \ge 3$.

\item [(1012):] $H(A,B,C)=2$ implies $A \ne C$. 
        $\LHS=4$ implies $H(A|Z,Y)=0$ or $H(C|A,X)=0$.

        \textbf{Case (1)}: 
        Suppose $H(C|A,X)=0$. Then $H(X) \ge 1$ and $X \ne A$ which imply $\RHS \ge 1$.
        Thus, $H(X|W,Y)=0$ which implies $H(W)+H(Y) \ge H(X)$, which implies $\RHS \ge 2$
        and therefore $H(A|Z,Y)=0$.
        We have $H(W|A,B)=0$ which implies $W$ is a subspace of $A$,
        which implies $H(Y) \ge 1$ and $Y \ne A$.
        Also, $H(Y|B,C)=0$ which implies $Y=C$ and therefore $H(Z) \ge 1$ and $Z \ne C$.
        Finally, $H(Z|W,C)=0$ which implies $H(W) \ge 1$ and therefore $\RHS \ge 4$.
        
        \textbf{Case (2)}: 
        Suppose $H(A|Z,Y)=0$. We know $H(Y|B,C)=0$, which implies $Y$ is a subspace of $C$
        which implies $H(Z) \ge 1$ and $Z \ne C$ and therefore $\RHS \ge 1$.
        Thus, $H(Z|W,C)=0$ which implies $H(W) \ge 1$ which implies
        $\RHS \ge 2$. So, $H(C|A,X)=0$ which implies $H(X) \ge 1$ and $X \ne A$
        and therefore $\RHS \ge 3$.
        Also, $H(W|A,B)=0$ which implies $W=A$. 
        Finally, $H(X|W,Y)=0$ which implies $H(Y) \ge 1$ and therefore $\RHS \ge 4$.
\end{itemize}
\end{proof}


\newpage
\subsection{A Linear Rank Inequality from the non-Fano Network}
\label{sec:NonFanoProof2}

\begin{theorem}
Let $A,B,C,W,X,Y,Z$ be subspaces of a finite-dimensional vector space $V$
over a scalar field of even characteristic.
Then, the following linear rank inequality holds:
\begin{align}
&2H(A) + 3H(B) + 2H(C) \notag \\
&\ \ \le H(W) + H(X) + H(Y) + 3H(Z) \notag \\
&\ \ \  + 2H(A|Y,Z) + 3H(B|X,Z) + H(C|W,Z) \notag \\
&\ \ \  + 2H(W|A,B) + 4H(X|A,C) + 3H(Y|B,C) \notag \\
&\ \ \  + 6H(Z|A,B,C) + H(C|W,X,Y) \notag \\
&\ \ \  + 7( H(A) + H(B) + H(C) - H(A,B,C) ).
\label{evenLRI}
\end{align}
\label{thm:LinearRankInequality2}
\end{theorem}

\begin{proof}
We will use the non-Fano network in Figure~\ref{fig:non-Fano-network},
derived in~\cite{Dougherty-Freiling-Zeger07-NetworksMatroidsNonShannon},
from the non-Fano matroid,
to help guide the proof.
By Lemma~\ref{lem:1},
there exist linear functions
\begin{align*}
f_1:    W &\to A & f_2:    W &\to B\\
f_3:    X &\to A & f_4:    X &\to C\\
f_5:    Y &\to B & f_6:    Y &\to C\\
f_7:    Z &\to A & f_8:    Z &\to B & f_9:    Z &\to C \\
f_{10}: C &\to W & f_{11}: C &\to Z\\
f_{12}: B &\to X & f_{13}: B &\to Z\\
f_{14}: A &\to Y & f_{15}: A &\to Z\\
f_{16}: C &\to W & f_{17}: C &\to X & f_{18}: C &\to Y
\end{align*}
such that
\begin{align}
f_1    + f_2    &= I \mbox{ on a subspace $W'$ of $W$ with } \Codim{W}{W'} \le H(W|A,B)\label{eq:141}\\
f_3    + f_4    &= I \mbox{ on a subspace $X'$ of $X$ with } \Codim{X}{X'} \le H(X|A,C)\label{eq:142}\\
f_5    + f_6    &= I \mbox{ on a subspace $Y'$ of $Y$ with } \Codim{Y}{Y'} \le H(Y|B,C)\label{eq:143}\\ 
f_7 + f_8 + f_9 &= I \mbox{ on a subspace $Z'$ of $Z$ with } \Codim{Z}{Z'} \le H(Z|A,B,C)\label{eq:144}\\
f_{10} + f_{11} &= I \mbox{ on a subspace $C'$ of $C$ with } \Codim{C}{C'} \le H(C|W,Z)\label{eq:145}\\
f_{12} + f_{13} &= I \mbox{ on a subspace $B'$ of $B$ with } \Codim{B}{B'} \le H(B|X,Z)\label{eq:146}\\
f_{14} + f_{15} &= I \mbox{ on a subspace $A'$ of $A$ with } \Codim{A}{A'} \le H(A|Y,Z)\label{eq:147}\\
f_{16}+f_{17}+f_{18} &= I \mbox{ on a subspace $C''$ of $C$ with } \Codim{C}{C''} \le H(C|W,X,Y).\label{eq:147b}
\end{align}
Combining these, we get maps
\begin{align}
f_7 f_{15}: A \to A\label{eq:138}\\
f_5 f_{14} + f_8 f_{15}: A \to B\label{eq:139}\\
f_6 f_{14} + f_9 f_{15}: A \to C.\label{eq:140}
\end{align}
Note that
\begin{align*}
f_5 f_{14} + f_6 f_{14} &= f_{14} 
\mbox{ on the subspace } f_{14}^{-1} (Y') \mbox{ of } A\\
f_7 f_{15} + f_8 f_{15} + f_9 f_{15} &= f_{15} 
\mbox{ on the subspace } f_{15}^{-1} (Z') \mbox{ of } A
\end{align*}
so the sum of the functions in \eqref{eq:138}--\eqref{eq:140} 
is equal to $I$ on the subspace 
$$
A'' \doteq
A' \cap
f_{14}^{-1} (Y') \cap
f_{15}^{-1} (Z')
$$
and we get
\begin{align*}
\Codim{A}{A''}
&\le
\Codim{A}{A'}
+ \Codim{A}{f_{14}^{-1} (Y')}
+ \Codim{A}{f_{15}^{-1} (Z')} &\Comment{Lemma~\ref{lem:2}}\\
&\le
\Codim{A}{A'} + \Codim{Y}{Y'} + \Codim{Z}{Z'} &\Comment{Lemma~\ref{lem:3}} \\
&\le
H(A|Y,Z) + H(Y|B,C) + H(Z|A,B,C). &\Comment{\eqref{eq:143},\eqref{eq:144},\eqref{eq:147}}
\end{align*}
Applying Lemma~\ref{lem:2b} to
$f_7 f_{15} - I$,
$f_5 f_{14} + f_8 f_{15}$,
and
$f_6 f_{14} + f_9 f_{15}$,
we get a subspace $\bar{A}$ of $A''$ such that
\begin{align}
\Codim{A}{\bar{A}} 
&= \Codim{A}{A''} + \Codim{A''}{\bar{A}}\notag\\
&\le \Delta_A \label{eq:130}\\
&\doteq H(A|Y,Z) + H(Y|B,C) + H(Z|A,B,C)\notag\\
&\ \ \ + H(A) + H(B) + H(C) - H(A,B,C)\label{eq:131}
\end{align}
on which
\begin{align}
f_7 f_{15} &= I\label{eq:137}\\
f_5 f_{14} + f_8 f_{15} &=0 \label{eq:137b}\\
f_6 f_{14} + f_9 f_{15} &= 0\label{eq:137c}.
\end{align}
Similarly,
we get a subspace $\bar{B}$ of $B$ such that
\begin{align}
\Codim{B}{\bar{B}} &\le \Delta_B\label{eq:132}\\
&\doteq H(B|X,Z) + H(X|A,C) + H(Z|A,B,C) \notag\\
&\ \ \ + H(A) + H(B) + H(C) - H(A,B,C)\label{eq:133}
\end{align}
on which
\begin{align}
f_8 f_{13} &= I\label{eq:136}\\
f_3 f_{12} + f_7 f_{13} &=0\label{eq:136bb}\\
f_4 f_{12}+ f_9 f_{13} &= 0 \label{eq:136c}
\end{align}
and a subspace $\bar{C}$ of $C$ such that
\begin{align}
\Codim{C}{\bar{C}} &\le \Delta_C\label{eq:134}\\
&\doteq H(C|W,Z) + H(W|A,B) + H(Z|A,B,C)\notag\\
&\ \ \ + H(A) + H(B) + H(C) - H(A,B,C)\label{eq:135}
\end{align}
on which
\begin{align}
f_9 f_{11} &= I\label{eq:136b}\\
f_1 f_{10} + f_7 f_{11} &=0\label{eq:136d}\\
f_2 f_{10} + f_8 f_{11} &= 0\label{eq:136e}
\end{align}
and a subspace $\hat{C}$ of $C$ such that
\begin{align}
\Codim{C}{\hat{C}} &\le \hat{\Delta}_C\label{eq:134b}\\
&\doteq H(C|W,X,Y) + H(W|A,B) + H(X|A,C) + H(Y|B,C)\notag\\
&\ \ \  + H(A) + H(B) + H(C) - H(A,B,C)\label{eq:135b}
\end{align}
on which
\begin{align}
f_4 f_{17} + f_6 f_{18} &= I\label{eq:a1}\\
f_1 f_{16} + f_3 f_{17} &=0\label{eq:b1}\\
f_2 f_{16} + f_5 f_{18} &= 0.\label{eq:c1}
\end{align}
Define the following subspaces of $Z$:
\begin{align*}
A^* &= f_{15}(\bar{A})\\
B^* &= f_{13}(\bar{B})\\
C^* &= f_{11}(\bar{C}).
\end{align*}
By \eqref{eq:137},
the restriction maps
$\Restriction{f_{15}}{\bar{A}} : \bar{A} \to A^*$
and
$\Restriction{f_7}{A^*} : A^* \to \bar{A}$
are inverses of each other, and hence are injective.
Similarly, by \eqref{eq:136},
$\Restriction{f_8}{B^*}$
is the inverse of 
$\Restriction{f_{13}}{\bar{B}}$
and, by by \eqref{eq:136b},
$\Restriction{f_9}{C^*}$
is the inverse of 
$\Restriction{f_{11}}{\bar{C}}$,
so these are all injective.
In particular,
\begin{align}
\Dim{A^*} &= \Dim{\bar{A}}\label{eq:52}\\
\Dim{B^*} &= \Dim{\bar{B}}\label{eq:53}\\
\Dim{C^*} &= \Dim{\bar{C}}.\label{eq:54}
\end{align}

Now let
$$A^{**} = f_7 ( A^* \cap B^* ) \subseteq \bar{A}.$$
Then $f_{15}$ is injective on $A^{**}$
and $f_{15}(A^{**}) = A^* \cap B^*$,
so $f_8 f_{15}$ is injective on $A^{**}$.
But $f_5 f_{14} + f_8 f_{15} = 0$ on $\bar{A}$,
so $f_5 f_{14}$ is injective on $A^{**}$,
and hence so is $f_{14}$.
This gives 
\begin{align}
\Dim{f_{14} A^{**}} = \Dim{A^{**}} = \Dim{ A^* \cap B^*}.\label{eq:56}
\end{align}
Similarly, let
$$B^{**} = f_8 ( A^* \cap B^* ) \subseteq \bar{B};$$
then $f_7 f_{13}$ is injective on $B^{**}$ and
$f_3 f_{12} + f_7 f_{13} = 0$ on $B^{**}$,
so $f_{12}$ is injective on $B^{**}$ and
\begin{align}
\Dim{f_{12} B^{**}} = \Dim{B^{**}} = \Dim{ A^* \cap B^*}.\label{eq:55}
\end{align}
And let 
$$C^{**} = f_9 ( B^* \cap C^* ) \subseteq \bar{C};$$
then $f_8 f_{11}$ is injective on $C^{**}$ and
$f_2 f_{10} + f_8 f_{11} = 0$ on $C^{**}$,
so $f_{10}$ is injective on $C^{**}$ and
\begin{align}
\Dim{f_{10} C^{**}} = \Dim{C^{**}} = \Dim{ B^* \cap C^*}. \label{eq:51}
\end{align}

Let us define the following subspaces of $C$:
\begin{align}
S_1 &= \{ u\in C: f_{16}u \in f_{10} C^{**} \}\notag\\
S_2 &= \{ u\in C: f_{17}u \in f_{12} B^{**} \}\notag\\
S_3 &= \{ u\in C: f_{18}u \in f_{14} A^{**} \}\notag\\
S &= \hat{C} \cap S_1 \cap S_2 \cap S_3.\label{eq:S2}
\end{align}
Then we have the following:
\begin{align}
\Codim{C}{S_1} 
&\le \Codim{W}{f_{10}C^{**}} &\Comment{Lemma~\ref{lem:3}}\notag\\
&= \Dim{W} - \Dim{B^* \cap C^*} &\Comment{\eqref{eq:51}}\notag\\
&= \Codim{Z}{B^* \cap C^*} + \Dim{W} - \Dim{Z} \notag\\
&\le \Codim{Z}{B^*} + \Codim{Z}{C^*} + \Dim{W} - \Dim{Z} &\Comment{Lemma~\ref{lem:2}}\notag\\ 
&= \Codim{B}{\bar{B}} + \Codim{C}{\bar{C}}\notag\\
&\ \ \  + \Dim{W} + \Dim{Z} - \Dim{B} - \Dim{C} &\Comment{\eqref{eq:53},\eqref{eq:54}}\notag\\
&\le \Delta_B + \Delta_C + H(W) + H(Z) - H(B) - H(C) &\Comment{\eqref{eq:132},\eqref{eq:134}}\label{eq:126}\\
\Codim{C}{S_2} 
&\le \Codim{X}{f_{12}B^{**}}  &\Comment{Lemma~\ref{lem:3}}\notag\\
&= \Dim{X} - \Dim{A^* \cap B^*} &\Comment{\eqref{eq:55}}\notag\\
&= \Codim{Z}{A^* \cap B^*} + \Dim{X} - \Dim{Z} \notag\\
&\le \Codim{Z}{A^*} + \Codim{Z}{B^*} + \Dim{X} - \Dim{Z} &\Comment{Lemma~\ref{lem:2}}\notag\\ 
&= \Codim{A}{\bar{A}} + \Codim{B}{\bar{B}}\notag\\
&\ \ \  + \Dim{X} + \Dim{Z} - \Dim{A} - \Dim{B} &\Comment{\eqref{eq:52},\eqref{eq:53}}\notag\\
&\le \Delta_A + \Delta_B + H(X) + H(Z) - H(A) - H(B) &\Comment{\eqref{eq:130},\eqref{eq:132}}\label{eq:126b}\\
\Codim{C}{S_3} 
&\le \Codim{Y}{f_{14}A^{**}} &\Comment{Lemma~\ref{lem:3}}\notag\\
&= \Dim{Y} - \Dim{A^* \cap B^*} &\Comment{\eqref{eq:56}}\notag\\
&= \Codim{Z}{A^* \cap B^*} + \Dim{Y} - \Dim{Z} \notag\\
&\le \Codim{Z}{A^*} + \Codim{Z}{B^*} + \Dim{Y} - \Dim{Z} &\Comment{Lemma~\ref{lem:2}}\notag\\ 
&= \Codim{A}{\bar{A}} + \Codim{B}{\bar{B}}\notag\\
&\ \ \  + \Dim{Y} + \Dim{Z} - \Dim{A} - \Dim{B} &\Comment{\eqref{eq:52},\eqref{eq:53}}\notag\\
&\le \Delta_A + \Delta_B + H(Y) + H(Z) - H(A) - H(B). &\Comment{\eqref{eq:130},\eqref{eq:132}}\label{eq:126c}
\end{align}
Suppose $t\in S$.
Then there exist 
$a\in A^{**}$,
$b\in B^{**}$,
and
$c\in C^{**}$
such that
$f_{14} a = f_{18} t$,
$f_{12} b = f_{17} t$,
and
$f_{10} c = f_{16} t$.
Since $t\in \hat{C}$, we have
from (\eqref{eq:a1})--(\eqref{eq:c1}) that
\begin{align*}
f_1 f_{16} t + f_3 f_{17} t &= 0\\
f_2 f_{16} t + f_5 f_{18} t &= 0\\
f_4 f_{17} t + f_6 f_{18} t &= t
\end{align*}
which gives
\begin{align}
f_1 f_{10} c + f_3 f_{12} b &= 0\label{eq:201}\\
f_2 f_{10} c + f_5 f_{14} a &= 0\label{eq:202}\\
f_4 f_{12} b + f_6 f_{14} a &= t.\label{eq:203}
\end{align}
But we also have
\begin{align}
f_5 f_{14} a + f_8 f_{15} a &= 0 &\Comment{\eqref{eq:137b}} \label{eq:204}\\
f_6 f_{14} a + f_9 f_{15} a &= 0 &\Comment{\eqref{eq:137c}} \label{eq:205}\\
f_3 f_{12} b + f_7 f_{13} b &= 0 &\Comment{\eqref{eq:136bb}}\label{eq:206}\\
f_4 f_{12} b + f_9 f_{13} b &= 0 &\Comment{\eqref{eq:136c}} \label{eq:207}\\
f_1 f_{10} c + f_7 f_{11} c &= 0 &\Comment{\eqref{eq:136d}} \label{eq:208}\\
f_2 f_{10} c + f_8 f_{11} c &= 0 &\Comment{\eqref{eq:136e}} \label{eq:209}
\end{align}
so
\begin{align}
f_7 f_{11} c + f_7 f_{13} b &= 0   &\Comment{\eqref{eq:201},\eqref{eq:206},\eqref{eq:208}}\label{eq:58}\\
f_8 f_{11} c + f_8 f_{15} a &= 0   &\Comment{\eqref{eq:202},\eqref{eq:204}\eqref{eq:209}}\label{eq:57}\\
f_9 f_{13} b + f_9 f_{15} a &= -t. &\Comment{\eqref{eq:203},\eqref{eq:205},\eqref{eq:207}}
\end{align}
Since $f_{11}c$ and $f_{15}a$ are both in $B^*$,
and $f_8$ is injective on $B^*$,
we get from \eqref{eq:57} that
$f_{11} c = -f_{15} a$.
This implies that $f_{11} c$ is also in $A^*$,
and since $f_{13}b \in A^*$ and 
$f_7$ is injective on $A^*$,
we get from \eqref{eq:58} that 
$f_{11} c = - f_{13}b$ and hence
$f_{15} a = f_{13} b$.

Hence, since the field has characteristic $2$,
we have
\begin{align*}
t &= -( f_9 f_{13} b + f_9 f_{15} a)\\
  &= -( f_9 f_{13} b + f_9 f_{13} b)\\
  &= 0.
\end{align*}
Since the choice of $t$ was arbitrary, this implies $S=\{0\}$,
and therefore
\begin{align*}
H(C) &= \Codim{C}{S}\\
&\le \Codim{C}{\hat{C}} + \sum_{i=1}^3 \Codim{C}{S_i}&\Comment{\eqref{eq:S2}, Lemma~\ref{lem:2}}\\
&\le \hat{\Delta}_C + 2 \Delta_A + 3 \Delta_B + \Delta_C\\
&\ \ \  + H(W) + H(X) + H(Y) + 3H(Z)\\
&\ \ \  - 2H(A) - 3H(B) - H(C) &\Comment{\eqref{eq:134b},\eqref{eq:126},\eqref{eq:126b},\eqref{eq:126c}}.
\end{align*}
The result then follows from
\eqref{eq:131},
\eqref{eq:133},
\eqref{eq:135}.
and
\eqref{eq:135b}.
\end{proof}

In the context of the non-Fano network, all of the compound terms at
the end of inequality~\eqref{evenLRI} are zero, so this inequality
directly implies inequality~\eqref{eq:23b}.

\begin{theorem}
The linear rank inequality in Theorem~\ref{thm:LinearRankInequality2}
holds for any scalar field if $\Dim{V} \le 2$,
but may not hold if the scalar field has odd characteristic
and $\Dim{V} \ge 3$.

\end{theorem}

\begin{proof}
In $V=GF(p)^3$ for any odd prime $p$, define the following subspaces of $V$:
\begin{align*}
A &= \Span{(1,0,0)}\\
B &= \Span{(0,1,0)}\\
C &= \Span{(0,0,1)}\\
W &= \Span{(1,1,0)}\\
X &= \Span{(1,0,1)}\\
Y &= \Span{(0,1,1)}\\
Z &= \Span{(1,1,1)}
\end{align*}
It is easily verified that the inequality in Theorem~\ref{thm:LinearRankInequality2}
is not satisfied in this case.

To show that the inequality indeed holds if $\Dim{V} \le 2$,
one can again show that the inequality becomes a Shannon inequality
under the assumption that $H(A) = 0$, or under the assumption
$H(B|A) = 0$, or under the assumption $H(C|A,B) = 0$.  If all three
of these assumptions fail, then we must have
\begin{align}
\Dim{V} \ge H(A,B,C) > H(A,B) > H(A) > 0
\end{align}
and hence $\Dim{V} \ge 3$.
Or one can give a case-by-case direct argument.
\end{proof}


\clearpage

\renewcommand{\baselinestretch}{1.0}

\end{document}